\documentclass[twocolumn]{autart}    

\usepackage[pdftex]{graphicx}                               

\usepackage{acronym}

\usepackage{amssymb}
\usepackage{mathtools}
\usepackage{paralist}
\usepackage{bm}
\usepackage{algorithm}
\usepackage[noend]{algpseudocode}
\usepackage{subcaption}
\usepackage{url}
\usepackage{booktabs}
\let\theoremstyle\relax
\usepackage{amsthm}
\usepackage{graphicx}

  \newcommand{\dist}[1]{\mathcal{D}(#1)}

\newcommand{\reals}{\mathbf{R}}
 
\newtheorem{definition}{Definition}
\newtheorem{proposition}{Proposition}

\newtheorem{problem}{Problem}

\newtheorem{remark}{Remark}

\newtheorem{lemma}{Lemma}

  \acrodef{mdp}[MDP]{Markov decision process} 

 \newcommand{\abs}[1]{\lvert #1 \rvert}
 
 \DeclareMathOperator*{\optmin}{\textrm{minimize}}
 \DeclareMathOperator*{\optmax}{\textrm{maximize}}

 \DeclareMathOperator*{\optst}{\textrm{subject to}}

 \newcommand{\argmax}{\mathop{\mathrm{argmax}}}

  \newcommand{\Expect}{\mathbb{E}}

\newcommand{\bbP}{\mathbb{P}}

 \newcommand{\probs}{\mathbb{P}}

\acrodef{mdp}[MDP]{Markov decision process}
\acrodef{hmm}[HMM]{hidden Markov model}
\acrodef{pomdp}[POMDP]{partially observable Markov decision process}

\newcommand{\Qstar}{Q^{\ast}}

 \newcommand{\Ro}{\bar{R}}  

\newcommand{\matA}{{\mathbf{A}}}
\newcommand{\matT}{{\mathbf{T}}}
\newcommand{\matO}{{\mathbf{O}}}

\begin{document}

\begin{frontmatter}

\title{Active Inference via Incentive Design: Eliciting Divergent Behaviors under Partial Observability\thanksref{footnoteinfo}} 

\thanks[footnoteinfo]{This paper was not presented at any IFAC 
meeting. Corresponding author M.~T.~Cicero. Tel. +XXXIX-VI-mmmxxi. 
Fax +XXXIX-VI-mmmxxv.}

\author[Gainesville]{Xinyi Wei}\ead{weixinyi@ufl.edu},    
\author[Gainesville]{Chongyang Shi}\ead{c.shi@ufl.edu},               
\author[Chicago]{Shuo Han}\ead{hanshuo@uic.edu}  
\author[Adelphi]{Ahmed Hemida}\ead{ahmed.h.hemida.ctr@army.mil}
\author[Adelphi]{Charles A. Kamhoua}\ead{charles.a.kamhoua.civ@army.mil}
\author[Gainesville]{Jie Fu}\ead{fujie@ufl.edu}

\address[Gainesville]{Gainesville, FL, USA}  
\address[Chicago]{Chicago, IL, USA}             
\address[Adelphi]{Adelphi, MD, USA}        


\begin{abstract}                          
We introduce Active Inference through Incentive Design, a framework in which a leader offers strategic incentive to elicit diverging behaviors from followers of different types, thereby improving intent inference under imperfect observations.  We analyze two settings: Case 1: a perfectly observable leader with two homogeneous followers with different original/unmodified reward functions,   under the assumption that followers' best responses induce   Gibbs trajectory distributions in their respective Markov decision processes, we show that this inference problem can be reformulated as maximizing the difference in expected cumulative return, under a modified reward for one follower, between the best-response policies of the two followers; Case 2 considers a  partially observable leader who aims to distinguish among heterogeneous followers differing in dynamics and/or rewards. The leader balances incentive cost against information gain, measured by the posterior entropy over follower types, consider the follower's best responses to the incentive design. We develop an efficient first-order gradient-based algorithm   to optimize the information-theoretic inference objective. Experiments in stochastic gridworld environments demonstrate significant improvements in both inference accuracy and efficiency over unincentivized baselines.
\end{abstract}

\end{frontmatter}

\section{Introduction}
The ability to recognize an agent’s actions, plans, and goals is fundamental in many AI tasks~\cite{van2021activity}. Goal recognition refers to the task of inferring an actor’s goal based on a sequence of observed behaviors~\cite{meneguzzi2021survey,mirsky2021introduction,chiari2023goal,amado2022goal}. Activity recognition, also known as behavior recognition, aims to identify the specific activity that an agent is performing~\cite{vrigkas2015review,jobanputra2019human,poppe2010survey}. Plan recognition focuses on understanding both the goal of autonomous agents and the set of actions that have been or will be executed to achieve that goal, given a series of observed actions~\cite{carberry2001techniques,sengupta2017radar}.

While extensive research has addressed the problems of inferring the goal or reward function of an agent from a hypothesis space, a complementary question remains   under-explored: \emph{How can we design environments that amplify differences in observed agent behavior?} This question is closely related to \textit{active inference}, a framework in which agents act to influence observed information in order to minimize uncertainty about latent or unknown states.
For instance, a teacher may wish to design an exam that more clearly distinguishes students of different proficiency levels. Similarly, an organization may seek to construct monitoring environments that make it easier to distinguish between normal users (e.g., employees) and malicious intruders.

We study this active inference problem through the lens of incentive design. Incentive design~\cite{boltonContractTheory2005,ho1981information},   also referred to as the \textit{principal-agent} or \textit{leader-follower} game, studies the problem where  a  leader aims to optimize system performance while anticipating the responses of multiple followers. 
 This work connects active inference and incentive design as follows. Consider a setting where a leader agent partially observes the behavior of a follower agent  and aims to infer the follower’s unknown type. The follower’s interaction with the dynamic environment is modeled as a Markov decision process (MDP). Different follower types are characterized by distinct dynamics, discount factors, or reward functions, and each follower acts optimally to maximize its own reward. While each follower has perfect observations within its own MDP, the leader receives only imperfect and noisy observations of the follower’s activities.
  
  To improve inference accuracy within a finite time horizon, we propose that the leader strategically designs an incentive policy, offering side payment within the environment, to amplify the divergence in observable behaviors across follower types. 
These incentives, however, incur a   cost to the leader's own payoff. The leader calibrates the incentive policy to minimize both uncertainty about the follower's type and  the cost of providing incentives. This trade-off gives rise to a   bilevel optimization problem formulation where the leader jointly optimizes inference accuracy and cost of incentive, anticipating the best responses from followers.

\paragraph*{Related work} Incentive design where the follower's planning problem is modeled as an \ac{mdp}   , has appeared in the literature under several related terminology, including  ``reward tempering''~\cite{everitt2021reward,skalse2022defining}, ``model design''~\cite{chenAdaptiveModelDesign2022}, ``behavior modification''~\cite{savasComplexitySequentialIncentive2022}, ``bilevel reinforcement learning'' \cite{thomaContextualBilevelReinforcement2024} and ``environment design''~\cite{zhangValuebasedPolicyTeaching2008,zhangGeneralApproachEnvironment2009}.  
Traditionally, both leader and follower's objectives are to maximize their respective total rewards. Adaptive incentive design~\cite{ratliff2020adaptive,ma2024adaptive} extends this framework to settings where the follower’s reward function is unknown. Instead of explicitly inferring the follower’s rewards from interactions, these algorithms dynamically adjust the leader’s policy in response to observed follower behavior.

In a classical leader-follower game, the leader and followers
both aim to optimize their cumulative rewards. However, in our problem, the leader’s objective is information-theoretic
rather than reward-based. 
Information-theoretic metrics are widely used in   active inference problems but not for incentive design. The work \cite{egorovTargetSurveillanceAdversarial2016,Araya2010Neurips} formulate a \ac{pomdp} with a belief-based reward. For instance, in target surveillance, a patrolling team is rewarded for reducing uncertainty in its belief about an intruder’s position or state. Similarly, in intent inference, the authors \cite{Shendeep2019} utilized  the negative entropy of the belief over an opponent’s intent as a reward, maximizing the total reward to enhance active inference. Also, the work \cite{Shi2024} uses Shannon conditional entropy as an information leakage measure to solve the active inference problem in which an agent can actively query sensors to infer an unknown variable in a hidden Markov model. In \cite{zeng2018inverse}, the authors use inverse reinforcement learning   to infer the actor’s rewards. However, no existing method explicitly addresses the recognition of agents with different transition dynamics or a mixture of different dynamics and rewards. This distinction is important in practical applications. For instance, when an engineer designs a test to better differentiate human operators of varying proficiency levels in tele-operating robots, these operators typically share the same reward for completing tasks, but differences in their tele-operating abilities lead to distinct  interaction dynamics.

\paragraph*{Our contribution}
Our research has the following key contributions:
\begin{itemize}
    \item We introduce an incentive design framework for active inference. The leader selects incentives to maximize the statistical distinguishability between follower types: in the fully observable setting, this is captured by the KL divergence between trajectory distributions; under partial observations, it is measured by the reduction in the leader's posterior entropy over follower types.
    \item  We show that, when the leader has  perfect observation and the followers' best responses induces Gibbs trajectory distributions, maximizing the distinguishability between two followers is equivalent to maximizing the \textit{cross-reward value gap}, defined as the
difference in expected cumulative return under one followers' modified reward (original reward plus incentive) between
the two followers' best-response policies.
    \item  In the partially observable setting, we show that the leader's inference problem can be formulated as a bilevel optimization over an augmented MDP, where the follower's unknown type enters as a hidden state variable. This augmentation transforms the inference problem into the problem of estimating a hidden variable from a partially observable stochastic process. When each follower deploys a softmax-optimal policy, the bilevel problem reduces to a single-level optimization over the leader's incentives, which can be solved efficiently by first-order gradient descent, leveraging observable operator models \cite{jaeger2000observableoperator} for   accurate and   efficient gradient computation. 
    \item Finally, with experimental validation, we demonstrate  the effectiveness of our method in minimizing the leader’s uncertainty when all follower types respond optimally to the given incentive.
\end{itemize}

\section{Preliminaries and problem formulation}

\paragraph*{Notation} The set of real numbers is denoted by $\reals$. Given a finite set $S$, $\dist{S}$ denotes the set of all distributions over $S$.

The active inference problem involves a leader-follower framework where a leader can offer side payments to supplement the followers' original rewards, and the followers always take the best response. The leader can partially observe the followers' behaviors and aims to minimize uncertainty about their types, balancing against the cost of incentive.

\paragraph*{The follower's model} A single-agent \ac{mdp} is defined as:
$
    M = (S,A, P, \mu, \gamma, R),
$ where $S$ is a set of states, $A$ is a set of actions, $P\colon S \times A \rightarrow \dist{S}$ is a probabilistic transition function such that $P({s}'|{s}, {a})$ is the probability of reaching state ${s}'$ given action ${a}$ being taken at state ${s}$, $\mu \in \dist{S}$ is an initial state distribution, and $ R \colon S \times A \rightarrow \reals$ denotes the original reward function for the agent without any incentive where $ R(s, a)$ is the reward received by the follower for taking action $a$ in state $s$.

\paragraph*{Different types of followers}
We consider the interaction between a leader and a finite set $\mathcal{T}$ of followers with different types. 
For each type $i\in \mathcal{T}$,    follower $i$'s planning problem is modeled by an \ac{mdp}  $M_i = (S,A, P_i, \mu_i, \gamma_i, R_i)  $. Different types can have different initial state distributions, reward functions, discounting factors, or transition dynamics. Without loss of generality, we assume that the followers have the same state and action spaces.

\paragraph*{Incentives as side payments} 
The leader can allocate incentive(side payments) to the environment. A incentive is a function: $x \colon S \times A \to \reals_+$. Specifically, $x(s, a)$ is the additional non-negative reward that the leader offers to each follower when the follower takes action $a$ in state $s$. $x$ can be viewed as a vector in $R_+^{|S\times A|}$, which we denote by $\mathcal{X}$.

For each follower $i\in \mathcal{T}$, 
given a incentive vector $x=[x_1,\ldots, x_n]$, follower $i$'s modified  reward function $R_i( x)$  is defined as follows: For all $(s,a)\in S\times A$ and $i \in \mathcal{T}$,
    \begin{equation}
        \label{eq:sidepayment-reward-i}
         R_i(s, a; 
        x_i ) = R_i(s,a)+  x_i (s,a) .  
    \end{equation}
Treating the reward $R_i$ as a vector, we also write $ R_i+x_i$ for the modified reward function.
    
For each $i \in \mathcal{T}$,  follower $i$'s planning problem with incentive $x$ is thus modeled as the   \ac{mdp}:
$M_i(x_i) = (S,A, P_i, \mu_i, \gamma_i, \Ro_i(x_i)).$
Given the \ac{mdp} $M_i( x_i)$ and a Markov policy $\pi: S\rightarrow \dist{A}$, let $V_i(\mu_i, \pi_i; \Ro_i(x_i))$ denote the value   (total discounted reward) evaluated for policy $\pi$. Let $V_i^\ast(\mu_i; \Ro_i(x_i))$ be the optimal value function.

\begin{remark} 
	\label{rem:public-incentive}
	The subscript in $x_i$ should not be read as the leader choosing a
	different incentive for each follower type. The leader designs and
	deploys a single, \emph{public} incentive function $x\colon S\times A \to
	\reals_+$ into the environment, and this same function is offered to
	whichever follower is actually present.  However, the follower may have different private value for the public incentive, and that private value is $x_i = g_i(x)$. A special case is to enforce $x_i=x$ for all $i$.
	  the one
	shared function $x$ enters follower $i$'s own (type-specific) reward
	$R_i$. Where the theoretical
	analysis in Section~\ref{se:Hypergradient-Derivation} treats a pair
	$(x_1,x_2)$ as (formally) separate optimization variables, this is a
	theoretical device to derive a closed-form characterization; the
	identical-incentive case required for a genuine inference problem is
	recovered as the special instance $x_1=x_2=x$, under which the
	corresponding lemma and its proof continue to hold unchanged.
\end{remark}

\paragraph*{Leader's   observation functions}
We consider two classes of leaders: A perfect observable leader who can observe the state-action trajectories generated from the followers' policies in the followers \ac{mdp}s. A partially observable leader has   observation functions   defined as $E_i: S\rightarrow \dist{\mathcal{O}}$, for each $i\in \mathcal{T}$. That is, the leader only partially observes the state of the follower at each time step. 

An informal problem statement is given  below.
\begin{problem}
Consider that the leader can incentivize diverging optimal behaviors from different followers by providing incentives, which come at a cost to the leader. The question is: How should the leader design the incentive  to balance the trade-off between reducing uncertainty about agent types and the associated cost? 
\end{problem}



\section{Main Results}

Our main results are twofold. First, for a special class of active
inference by incentive design with a \textbf{perfectly observable} leader
and two followers sharing the same dynamics but different rewards, we show
that if each follower's best response induces a Gibbs trajectory
distribution, then the leader's inference problem can be reformulated as
maximizing the performance discrepancy between the two followers'
best-response policies, evaluated under one follower's modified reward.
Second, for the general case of a \textbf{partially observable} leader,
we present an algorithm based on hypergradient descent to determine a
locally optimal incentive design for active inference in scenarios
involving more than two follower types.

\subsection{Special case: A perfectly observable leader and two followers with different rewards}
\label{subsec:full-obs}

We consider a special case of two followers sharing the same dynamics
$M = (S, A, P, \mu)$ but with distinct reward functions $R_1$ and $R_2$.
The leader's goal is to design incentive  $x = (x, x)$ that maximally
discriminates between the two followers' behavior while minimizing the cost
of intervention. 

As noted in Remark~\ref{rem:public-incentive}, treating $x$ as
separate optimization variables in \eqref{eq:bi-level-two-follower-regular}
is a theoretical device for deriving a closed-form characterization; it
does not require the leader to know which follower is present.

To make the bilevel problem tractable, we restrict each follower's
best response to the class of   policies.
Specifically, each follower $i$ solves the KL-regularized objective, whose unique optimizer is characterized by
the following lemma.

\begin{lemma}[Gibbs Trajectory Distribution]
\label{lem:gibbs}
Given an \ac{mdp} $M = \langle S, A, P, \mu, R \rangle$ and a fixed reference
policy $\pi_0 \colon S \rightarrow \dist{A}$ (e.g., the uniform policy over
$A$), let
$$
  \zeta(\rho) \coloneqq \mu(s_0)\prod_{t=0}^{T-1}\pi_0(a_t\mid s_t)\,P(s_{t+1}\mid s_t,a_t)
$$
be the reference trajectory distribution induced jointly by $\pi_0$ and the
dynamics $P$. Because $\sum_{a}\pi_0(a\mid s)=1$ and $\sum_{s'}P(s'\mid
s,a)=1$ for all $s\in S$, $\zeta$ is a properly normalized probability
distribution over state-action trajectories. Let
$G(\rho) \coloneqq \sum_{t=0}^{T}\gamma^t R(s_t,a_t)$
denote the cumulative discounted return along trajectory $\rho$.
The KL-regularized optimal planning problem
\begin{equation}
  \label{eq:kl-variational}
  \max_{\mathbb{P}}\; V^\mathrm{reg}(\mu,R,\pi) \coloneqq
  \mathbb{E}_{\mathbb{P}}\bigl[G(\rho)\bigr]
  - \tau\, D_{\mathrm{KL}}(\mathbb{P} \| \zeta)
\end{equation}

admits a unique optimizer given by the Gibbs distribution
\begin{equation}
  \label{eq:gibbs_partition}
  \mathbb{P}^\ast(\rho)
  = \frac{\zeta(\rho)\exp\!\left(G(\rho)/\tau\right)}{Z},
  \quad
  Z \coloneqq \sum_{\rho}\zeta(\rho)\exp\!\left(G(\rho)/\tau\right).
\end{equation}

Moreover, when $\pi_0$ is uniform, the log-partition function coincides with
the KL-regularized optimal value function of \eqref{eq:soft-max}
evaluated at $\mu$, i.e., $\tau \log Z = V^{\mathrm{reg},\ast}(\mu, R)$; it is \emph{not},
in general, equal to $\mathbb{E}_{\mathbb{P}^\ast}[G(\rho)]$ (see Appendix
for the precise relation between the two).
\end{lemma}

\begin{proof}
  See Appendix \ref{app:gibbs}.
\end{proof}

This leads to the bilevel optimization problem:
\begin{equation}
  \label{eq:bi-level-two-follower-regular}
  \begin{alignedat}{2}
    &\optmax_{\bm{x} \in \mathcal{X}} \quad
    && D_{\mathrm{KL}}\!\left(M(\pi_1^\ast(x)) \,\big\|\, M(\pi_2^\ast(x))\right) \\
    &\optst
    && \pi_i^\ast(x) \in \argmax_{\pi \in \Pi}\;
       V_i^\mathrm{reg}\!\left(\mu,\, R_i(x),\, \pi\right), \quad i = 1, 2.
  \end{alignedat}
\end{equation}
where $D_{\mathrm{KL}}\!\left(M(\pi_1(x)) \,\big\|\, M(\pi_2(x))\right)$
    denotes the KL divergence between the trajectory distributions induced
    by the two followers' best-response policies $\pi_1(x)$ and
    $\pi_2(x)$ in the   MDP $M$. \(\mathcal{X}\) represent a constraint set that captures limitations on the incentive costs associated with deploying the recourse.

    $\lambda\, h(x)$ is the \emph{incentive cost}, where
    $h(x) \geq 0$ measures the resource   required to deploy
    the incentive $x$, and $\lambda > 0$ is a weight
    that controls the trade-off between maximizing distinguishability
    and minimizing the cost of intervention.

Under the KL-regularized best response of Lemma~\ref{lem:gibbs},
the bilevel problem \eqref{eq:bi-level-two-follower-regular} admits a
  reformulation, stated in the Lemma next.
 \begin{lemma}
 \label{lem:kl-value-gap}
 Given incentive $x$,
suppose both followers solve the KL-regularized optimal policy
\eqref{eq:kl-variational} for their modified rewards
$R_i + x$, $i=1,2$, relative to a common reference policy $\pi_0$.  
    When
     $\pi_i^\ast(x) \in \argmax_{\pi \in \Pi}\;
       V_i^\mathrm{reg}\!\left(\mu,\, R_i(x),\, \pi\right), \quad i = 1, 2.$
   \begin{equation}
   \label{eq:bilevel-gibbs-followers}
\begin{alignedat}{2}
	&D_{\mathrm{KL}}\!\left(M(\pi_1^\ast(x)) \,\big\|\, M(\pi_2^\ast(x))\right)  \\ =  & D_{\mathrm{KL}}(\bbP_1(x)\,\|\,\zeta) - D_{\mathrm{KL}}(\bbP_2(x)\,\|\,\zeta) \\
	& + \frac{1}{\tau}\Big[\mathbb{E}_{\bbP_2}[(R_2+x)(\rho)] - \mathbb{E}_{\bbP_1}[(R_2+x)(\rho)]\Big],
\end{alignedat}
\end{equation}
where $(R_2+x)(\rho) \coloneqq \sum_{t=0}^{T}\gamma^t(R_2+x)(s_t,a_t)$, and $\mathbb{P}_1(x)$, $\mathbb{P}_2(x)$ are the Gibbs trajectory distributions obtained from KL-regularized optimal policies $\{\pi_i^\ast(x) \}_{i=1,2}$, with $\zeta$ the reference trajectory distribution of Lemma~\ref{lem:gibbs}. This form expresses the objective entirely through the two followers' own trajectory expectations, without reference to any log-partition function.
\end{lemma}

\begin{proof}
    See Appendix ~\ref{app:kl-value-gap}.
\end{proof}

Intuitively, the right hand side in \eqref{eq:bilevel-gibbs-followers} includes two terms:  The second (bracketed) term is a cross-reward value gap: the leader benefits from raising follower 2's own expected return under its incentivized reward $R_2+x$ while making follower 1's trajectories score poorly under that same reward. The first term, $D_{\mathrm{KL}}(\bbP_1\,\|\,\zeta) - D_{\mathrm{KL}}(\bbP_2\,\|\,\zeta)$, rewards the leader for making follower 1's own behavior \emph{more} distinctive from the reference behavior $\zeta$ than follower 2's is, independently of any direct value comparison between the two. 


\begin{remark}[Interpretation of $D_{\mathrm{KL}}(\bbP_1\,\|\,\zeta)$]
The term $D_{\mathrm{KL}}(\bbP_1\,\|\,\zeta)$ decomposes into a discounted sum of per-step action divergences from the reference policy,
$$D_{\mathrm{KL}}(\bbP_1\,\|\,\zeta) = \Expect_{\bbP_1}\Big[\sum_{t=0}^{T}\gamma^t D_{\mathrm{KL}}\big(\pi_1(\cdot\mid s_t)\,\|\,\pi_0(\cdot\mid s_t)\big)\Big],$$
and can be read as a \emph{policy-commitment} cost, familiar from the KL-control/entropy-regularized RL literature~\cite{nachum2017bridging}: it measures how far follower 1's incentivized policy $\pi_1$ has moved from the passive baseline $\pi_0$ (e.g., uniform action selection), regardless of follower 2. If $x$ induces a sharply peaked, highly reward-driven policy, this term is large; if $x$ is weak (or $\tau$ is large) so that $\pi_1$ stays close to $\pi_0$, this term vanishes and follower 1's trajectories are difficult to tell apart from generic, undirected behavior --- including follower 2's. The same reading applies symmetrically to $D_{\mathrm{KL}}(\bbP_2\,\|\,\zeta)$ for follower 2 under $x$. This shows that the leader has an incentive to use $x$ (resp.\ $x$) not merely to misalign the two followers' returns under $R_2+x$, but also to make follower 1's own behavior more decisive and non-generic relative to follower 2's. In the regime where both $\pi_1$ and $\pi_2$ stay close to $\pi_0$, so that $D_{\mathrm{KL}}(\bbP_1\,\|\,\zeta) \approx D_{\mathrm{KL}}(\bbP_2\,\|\,\zeta) \approx 0$, the objective reduces to the cross-reward value gap $\frac{1}{\tau}\big[\mathbb{E}_{\bbP_2}[(R_2+x)(\rho)] - \mathbb{E}_{\bbP_1}[(R_2+x)(\rho)]\big]$.
\end{remark}

 From this result, we can solve the incentive design problem using a hypergradient ascent method for the objective function in \eqref{eq:bilevel-gibbs-followers}. In \cite{ma2024adaptive}, the algorithm for computing the hypergradient of a leader's value function $V(\mu, \pi^\ast(x))$ in an \ac{mdp}, where $\pi^\ast(x)$ is the best response, with respect to the incentive   $x$ is given. Thus, we can leverage the hypergradient ascent method in \cite{ma2024adaptive}  to   solve the inference by incentive design problem for two followers with KL-regularized optimal policies. Because $D_{\mathrm{KL}}(\bbP_1\,\|\,\zeta)$ and $D_{\mathrm{KL}}(\bbP_2\,\|\,\zeta)$ are themselves expectations under the entropy-regularized trajectory distributions induced by $x$  respectively, their hypergradients with respect to $x$  can each be computed with the same observable-operator machinery developed in Section~\ref{se:Hypergradient-Derivation} for the general partially observable case; the remaining cross-reward term $\mathbb{E}_{\bbP_2}[(R_2+x)(\rho)] - \mathbb{E}_{\bbP_1}[(R_2+x)(\rho)]$ is a standard policy-evaluation quantity whose hypergradient follows from the same observable-operator approach applied to the fixed reward $R_2+x$.

Let $\theta_i^\ast$ denote the parameters of follower $i$'s optimal policy, and let $\mathcal{L}$ denote the objective function in Problem \ref{eq:bilevel-gibbs-followers} without the cost function. Then, we have

\begin{equation}
\begin{aligned}
    \nabla_x \mathcal{L}(x)
    &=
    \left(
        \frac{\partial \theta_1^\ast(x)}{\partial x}
    \right)^{\!\top}
    \left.
    \nabla_\theta {\mathcal{L}}(\theta,x)
    \right|_{\theta=\theta_1^\ast(x)}
    \\
    &\quad+
    \left.
    \nabla_x {\mathcal{L}}(\theta,x)
    \right|_{\theta=\theta_1^\ast(x)}.
\end{aligned}
\end{equation}

\begin{proposition}
\label{prop:pL_px}
The partial derivative of the objective function in Problem \ref{eq:bilevel-gibbs-followers} takes the following form:

    \begin{equation}
\begin{aligned}
    \nabla_x \mathcal{L}(x) 
        = & \frac{1}{\tau^2}\mathrm{Cov}_{\mathbb{P}_1}\!\left[(R_1-R_2)(\rho),\; \phi(\rho)\right] 
         \\& + \frac{1}{\tau}\Big(
            \mathbb{E}_{\mathbb{P}_2}[\phi(\rho)] 
            - \mathbb{E}_{\mathbb{P}_1}[\phi(\rho)]
          \Big)    
\end{aligned}
\end{equation}
where $\phi(\rho) \coloneqq \sum_{t=0}^{T}\gamma^t \mathbf{e}_{(s_t,a_t)}$ 
is the discounted visitation vector of trajectory $\rho$.
\end{proposition}

\begin{proof}
    See Appendix \ref{app:pL_px}.
\end{proof}

\begin{proposition}
\label{prop:pL_ptheta}
    For the parameter vector $\theta = (\theta_1,\theta_2)$, and $\theta^\ast(x) = (\theta_1^\ast(x),\theta_2^\ast(x))$,

we have
\begin{equation}
    \left.
    \nabla_\theta{\mathcal{L}}(\theta,x)
    \right|_{\theta=\theta^\ast(x)}
    =
    \begin{pmatrix}
        \displaystyle
        \frac{1}{\tau}
        \operatorname{Cov}_{\mathbb{P}_1}\!\left[
            (R_1-R_2)(\rho),\,g_1^\ast(\rho)
        \right]
        \\[8pt]
        \mathbf{0}
    \end{pmatrix},
\end{equation}
where $g_1^\ast(\rho)
    \coloneqq
    \left.
    \nabla_{\theta_1}\log\mathbb{P}_{\theta_1}(\rho)
    \right|_{\theta_1=\theta_1^\ast(x)}$ is the trajectory score evaluated at follower $1$'s optimal
policy parameters.
\end{proposition}

\begin{proof}
    See Appendix ~\ref{app:pL_ptheta}.
\end{proof}

Under KL regularization relative to a reference policy \(\pi_0\), the optimal policy admits the closed-form Boltzmann policy~\cite{geist2019theory}.

\[
\pi^\ast(a\mid s)
=
\frac{
\pi_0(a\mid s)
\exp\!\left(Q^\ast(s,a)/\tau\right)
}{
\sum_{a'}
\pi_0(a'\mid s)
\exp\!\left(Q^\ast(s,a')/\tau\right)
}.
\]

Let \(\theta^\ast\) be chosen such that $\theta^\ast \coloneqq Q^\ast=\bigl[Q^\ast(s,a)\bigr]_{s\in S,\,a\in A}.$
Since the side payments are given by $x(s,a)=R(s,a)-\bar R(s,a),$ computing \(\frac{\partial \theta_1^\ast(x)}{\partial x}\) reduces to characterizing $\frac{\partial Q^\ast(s,a)}
{\partial R(\tilde{s},\tilde{a})}$, which is the discounted expected visitation of \((\tilde s,\tilde a)\) when starting from \((s,a)\) and subsequently following \(\pi^\ast\).

\subsection{General Case: A partially observable leader and multiple heterogeneous followers}
\label{subsec:partial-obs}
\label{se:Hypergradient-Derivation}
In this subsection, we consider the general case and present a hypergradient descent method for solving the locally optimal incentive with respect to the leader's inference objective. 
For simplicity,
in the following, \( R \) denotes a general reward function when presenting general statements, while \( \bar{R} \) denotes the reward function modified by incentive.

We first study how to measure the leader's uncertainty about the current follower's type after a finite sequence of observations, when the follower' policy is known and the leader observes the follower's behavior through imperfect observations.

\begin{definition}
 Given a collection of policies, referred to as a \emph{policy profile} of followers, $\bm{\pi} = [\pi_i]_{i\in \mathcal{T}}$, and the set of leader's observation functions for followers, $\{E_i, i \in \mathcal{T}\}$. The following \ac{hmm} can be constructed:
\[
\mathcal{M}(\bm{\pi})= \langle  
S\times \mathcal{T}, \mathbf{P}_{\bm{\pi}}, {\mathcal{O}}, {\mathcal{E}}, \mu_0 \rangle
\]

\begin{itemize}
    \item $S\times \mathcal{T} $ is the augmented state space. Each state $(s,i)$ includes a state in the follower's \ac{mdp} and a type of the follower.
    \item $\mathbf{P}_{\bm{\pi}}: S\times \mathcal{T}\rightarrow \dist{(S\times \mathcal{T})}$ is defined by
\begin{multline*}
    \mathbf{P}_{\bm{\pi}}((s', j)\mid (s,i)) \\
    =\begin{cases}
        \displaystyle\sum_{a\in A} P_i(s'\mid s, a)\,\pi_i(s, a) 
            & \text{if } i=j,\\[6pt]
        0 
            & \text{otherwise.}
    \end{cases}
\end{multline*}

In other words, at the state $(s,i)$, follower $i$ will take an action by following his policy $\pi_i ( s  ) $, and the type does not change.
    \item $\mu_0 \in \dist{(S\times \mathcal{T})}$ is the initial state distribution. For all $(s, i) \in S\times \mathcal{T}$, $$\mu_0(s,i) =  \mu_i(s)\probs(\mathbf{T}=i).$$
    where $\mathbf{T}$ is a random variable representing the estimated type of the follower. $\probs(\mathbf{T})$ is the prior distribution over possible types.
    \item $\mathcal{O}$ is a finite set of observations.
    \item ${\mathcal{E}}: S\times \mathcal{T}\rightarrow \dist{\mathcal {O}}$ is the observation function, defined by 
    $
    {\mathcal {E}}(o|(s,i)) = E_i(o|s) $, which is the probability of observing $o$ when agent $i$ at the state $s$.
 \end{itemize}
\end{definition}

Let $O_t$ denote a random variable representing the observation at time $t$, and let $o_t$ be a specific realization of this random variable. We denote the posterior estimate of the type $\mathbf{T}$ given an observation sequence $o_{0:T}$ under a policy profile $\bm{\pi}$ as $\probs_{\bm{\pi}}(\mathbf{T}|O_{0:T}= o_{0:T})$.

Next, we define the planning objective—Shannon conditional entropy. Entropy measures are particularly common for quantifying information leakage  \cite{khouzani2017leakage} and have been extensively studied in channel design problems, which can be interpreted as one-step decision-making problems.

\begin{definition} 
Let $Y\coloneqq O_{0:T}$.
The conditional Shannon entropy of the agent's type given the observations is defined by,
\begin{equation}
\begin{aligned}
&H (\mathbf{T}\mid Y, 
\mathcal{M}(\bm{\pi}) ) =\sum_{y\in \mathcal{Y}}\probs_{\bm{\pi}}(y)H(\mathbf{T}|Y=y, 
\mathcal{M}(\bm{\pi}))  \\ =& -  \sum_{i\in \mathcal{T} } \sum_{y\in \mathcal{Y}} \probs_{\bm{\pi}}(i,y)\log \probs_{\bm{\pi}}(i|y),
\end{aligned}
\end{equation} where $y$ is a sample observation, and $\mathcal{Y}$ is a set of all finite observation sequences of length $T$.
\end{definition}

Given different policy profiles $\bm{\pi}$ and $\bm{\pi}'$, the entropies $H (\mathbf{T}\mid O_{0:T}, \bm{\pi} ) $ and   $H (\mathbf{T}\mid O_{0:T}, \bm{\pi}') $ can be different. A policy profile makes it easier for the leader to recognize agent types if the conditional entropy is lower.

Based on this relation, we formulate the  incentive design problem where
the objective of the leader is to allocate  incentive ${x}$ to minimize uncertainty about the agent types give observations, balancing against the cost of incentive.

\begin{problem}[Active inference with incentive design] 
   \label{prob:bilevel}
  Assuming the followers always best responds to the incentive by taking an optimal policy in the follower's MDP $M_i(x)$, the leader's incentive design for active  inference problem is the following bi-level optimization problem:
  \begin{equation}
   \label{eq:bi-level}
	\begin{alignedat}{2}
      	& \optmin_{{x} \in\mathcal{X}} &  &  H(\mathbf{T}| O_{0:T}, \mathcal{M}(\bm{\pi}^\ast({x}))) + \lambda h({x}) \\\
		& \optst_{\phantom{x\in \reals_{+}^{\abs{S\times A}},\ \pi^\ast}} &\quad  & \pi_i^{\ast} (x)\in \argmax_{\pi\in \Pi}    V_i( \mu_{ i}, \pi ; \Ro_i(x_i) ), \forall i \in \mathcal{T}.
 	\end{alignedat}
  \end{equation}
 where $\bm{\pi}^\ast(x ) = [\pi_i^\ast(x)]_{i \in \mathcal{T}}$ is a policy profile consisting of the best responses of followers given incentive $x$,   $\lambda h(x)$ is the incentive cost.

\end{problem}

When  the follower employs an entropy-regularized optimal policy ~\cite{nachum2017bridging}, the best response  is unique. 
 The following reviews the entropy-regularized MDP solution for each single agent. We omit the index of a follower for clarity.


\paragraph*{Entropy-regularized optimal value/policy}
The optimal value function $V^\ast$ of the entropy-regularized \ac{mdp} with respect to the reward function $R$ satisfies the following entropy-regularized Bellman equation~\cite{nachum2017bridging}:  
\begin{equation}
\begin{aligned}
    V ^\ast(s, R ) = & 
    \tau \log \sum_{a \in A}\exp \{(R (s,a ) \\ & S+  
     \gamma  \Expect_{s'\sim P(\cdot\mid s,a)} V ^\ast(s', R  ))/\tau\}, \ \forall s\in S.
\end{aligned}
\label{eq:soft-max}
\end{equation} 
Let $\Qstar(R) \colon S\times A\rightarrow \reals$ 
    denote the optimal state-action value function (also called Q-function) of the entropy-regularized \ac{mdp} under reward $R$, we have
$
\Qstar(s,a, R )= R (s,a)+\Expect_{s'\sim P(\cdot |s,a)} V^\ast(s', R ).
$
For a fixed temperature parameter $\tau$, the optimal policy of the entropy-regularized \ac{mdp}  is uniquely defined ~\cite{nachum2017bridging}:
    \begin{equation}
        \label{eq:policy-softmax}
    \pi^\ast(s,a) = \frac{\exp(\Qstar(s,a, R)/\tau) }{\sum_{a'\in A} \exp(\Qstar(s,a',R)/\tau)}.
    \end{equation}

Then, the optimal policy of the entropy-regularized \ac{mdp} can be written succinctly as $\pi_{\Qstar(R )}$, where $\Qstar(R )$ is viewed as a vector in $\reals^{|S \times A|}$. 


Assuming follower's response is computed using entropy-regularized \ac{mdp},   the lower-level problem in the bilevel optimization problem in ~\eqref{eq:bi-level} has a unique solution $\pi_{\Qstar(R_i(x))}$.
   We can reduce the problem \ref{prob:bilevel} to the following single-level optimization problem.
     \begin{equation}
     \label{eq:single-level}
         \optmin_{{x} \in\mathcal{X}}   H(\mathbf{T}| O_{0:T}, \mathcal{M}(\bm{\pi}^\ast(x) ) + \lambda h({x}).
     \end{equation}
     where 
$\bm{\pi}^\ast(x) =[\pi_{\Qstar(R_i(x))}]_{i \in \mathcal{T}}$ is the policy profile.


Given the softmax parametrization, we will replace the notation of policy $\pi$ (or policy profile $\bm{\pi}$) with policy parameter $\theta$ (or policy parameter profile $\bm{\theta}\triangleq [\theta_i]_{i \in \mathcal{T}}$). Thus, $\mathcal{M}(\bm{\pi}_{\bm{\theta}}) \triangleq \mathcal{M}(\bm{\theta})$ and  $\probs_{\bm{\pi}_{\bm{\theta}}}\triangleq \probs_{\bm{\theta}}$. Similarly, we denote $\bm{\Qstar}(x )  \triangleq [\Qstar(\Ro_i(x_i))]_{i\in \mathcal{T}}$, and $\bar{\bm{R}}(x)  \triangleq [\Ro_i(x_i)]_{i\in \mathcal{T}}$.

Because   $\pi^\ast_i(x)$ is an implicit function of $x$, we consider using the gradient-descent method to find a stationary point    for the objective function in equation~\eqref{eq:bi-level}. 
Let's define $J_1(\bm{\theta}) = H(\mathbf{T}|Y;\mathcal{M}(\bm{\theta}))$ (where $Y\coloneqq O_{0:T}$). Note that $\theta_i$ should be ${\Qstar}(\Ro_i(x_i) ) $ in  the equation~\eqref{eq:single-level}.

Let $$J(x) \coloneqq J_1(\bm{\Qstar}(x ) ) + \lambda h(x).$$ 
Following the chain rule, the derivative of $J$ with respect to $x$ is given by
     \begin{equation}
     \label{eq:total-derivative}
D J(x) = DJ_1(\bm{\Qstar}(x )) \cdot D\bm{\Qstar}(x )  \cdot D \bar{\bm{R}}(x)  + \lambda D h(x).
     \end{equation}
    Because $h$ is given, $Dh(x)$ can be computed analytically. Similarly, $D\bar{\bm{R}}(x)$ can be computed analytically given the function $\Ro_i(x_i)$, for each $i \in \mathcal{T}$.

In the following, we show how to compute $DJ_1(\bm{\Qstar}(x) )$, i.e., $DJ_1(\bm{\theta} )|_{\bm{\theta}=\bm{\Qstar}(x)}$, and $D {\bm{\Qstar}(x)  }$, respectively.

\paragraph*{Computing $DJ_1(\bm{\theta})$} Computing $DJ_1(\bm{\theta})$ involves computing  the gradient of the conditional entropy $H(\mathbf{T}|Y;\mathcal{M}(\bm{\theta}))$ with respect to the  parameter $\bm{\theta}$ of policy profile, i.e., $DJ_1(\bm{\theta})=\nabla_{\bm{\theta}}H(\mathbf{T}|Y;\bm{\theta})^\intercal$.
    By using a trick that $\nabla_{\bm{\theta}} \probs_{\bm{\theta}}(y) = \probs_{\bm{\theta}}(y) \nabla_{\bm{\theta}} \log \probs_{\bm{\theta}}(y)$ and the property of conditional probability, we have
\begin{equation}
\begin{aligned}
\label{eq:HMM_gradient_entropy}
 &\nabla_{\bm{\theta}} H(\mathbf{T}|Y;\bm{\theta}) = - \sum_{y \in \mathcal{O}^T} \sum_{i \in \mathcal{T}} \Big[\nabla_{\bm{\theta}} \probs_{\bm{\theta}}(i, y) \log_2 \probs_{\bm{\theta}}(i | y)\\
&+  \probs_{\bm{\theta}}(i, y) \nabla_{\bm{\theta}}  \log_2 \probs_{\bm{\theta}}(i | y)\Big] \\
= & - \sum_{y \in \mathcal{O}^T} \sum_{i \in \mathcal{T}} \Big[\probs_{\bm{\theta}}(y) \nabla_{\bm{\theta}} \probs_{\bm{\theta}}(i| y) \log_2 \probs_{\bm{\theta}}(i | y)  \\ &  +  \probs_{\bm{\theta}}(i|y) \nabla_{\bm{\theta}} \probs_{\bm{\theta}}(y)\log_2 \probs_{\bm{\theta}}(i | y) 
+   \probs_{\bm{\theta}}(y)\frac{\nabla_{\bm{\theta}} \probs_{\bm{\theta}}(i | y)}{\ln 2}\Big].
\end{aligned}
\end{equation}

Given a prior distribution $\probs(\mathbf{T})$, we have $$\probs_{\bm{\theta}}(y) = \sum_{i \in \mathcal{T}}\probs(\mathbf{T}=i)\probs_{\bm{\theta}}(y|\mathbf{T}=i),$$
where $\probs_{\bm{\theta}}(y|\mathbf{T}=i) = \sum_{s\in S}\probs_{\bm{\theta}}(y|s,i)\mu_0(s|i) $ is the probability of generating the observation  sequence $y$ given the follower $i$ following a policy parameterized by ${\theta}_i$---the $i$-th component of policy profile parameters $\bm{\theta}$. We can calculate the gradient $\nabla_{\bm{\theta}} \probs_{\bm{\theta}}(\mathbf{T}=i|y)$ 
as
\begin{equation}
\begin{aligned}
\label{eq:HMM_gradient_P_zT_y}
&\nabla_{\bm{\theta}} \probs_{\bm{\theta}}(\mathbf{T} = i |y)
=  P(\mathbf{T}=i) \nabla_{\bm{\theta}} \frac{ \probs_{\bm{\theta}}(y|\mathbf{T}=i) }{\probs_{\bm{\theta}}(y)}\\
=  & P(\mathbf{T}=i)\Big[ \frac{\nabla_{\bm{\theta}}  \probs_{\bm{\theta}}(y|\mathbf{T}=i)}{\probs_{\bm{\theta}}(y)}  - \frac{ \probs_{\bm{\theta}}(y|\mathbf{T}=i)}{ \probs_{\bm{\theta}}^2(y)} \nabla_{\bm{\theta}} \probs_{\bm{\theta}}(y) \Big],
\end{aligned}
\end{equation} 
where $\nabla_{\bm{\theta}} \probs_{\bm{\theta}}(y) = \sum_{i \in \mathcal{T}}P(\mathbf{T}=i) \nabla_{\bm{\theta}}\probs_{\bm{\theta}}(y|\mathbf{T}=i)$.

To complete the gradient computation, we   need to determine $\nabla_{\bm{\theta}}\probs_{\bm{\theta}}(y|\mathbf{T}=i)$ for each type of follower. We can utilize the observable operators~\cite{jaeger2000observableoperator} to compute this gradient for the partially observable system induced by follower $i$'s policy.

\paragraph*{Observable operators for single-agent \ac{hmm}} 
We consider that a single-agent MDP $M= \langle S, A, P, \mu_0, R \rangle $, and a parameterized Markov policy $\pi_\theta$ induce a hidden Markov model
$
M_\theta = \langle  S, A, P_\theta, \mu_0, \mathcal{O}, E\rangle
$, where the state and observation sets are assumed to be $S =\{1,\ldots, N\}$ and $\mathcal{O}= \{1,\ldots, M\}$,  the transition probabilities under the policy are given by $P_\theta (s'|s)=\sum_{a\in A} P(s'|s,a)\pi_\theta(a|s)
$, and   $E: S\rightarrow \dist{\mathcal{O}}$ is an observation function for the leader.    

Let the random variable of state, observation, and control action in $M_\theta$, at time $t$ be denoted as $X_t, O_t, A_t$, respectively. Let $\matT^\theta \in \reals^{N \times N}$ be the transposed state transition matrix in the single-agent HMM $M_\theta$ with $$\matT^\theta_{i,j} = P_\theta(X_{t+1} = i|X_t = j).$$ 

Let $\matO \in \reals^{M \times N}$ be the observation probability matrix with $\matO_{o,j} =E( o|j)$ for each $o\in \mathcal{O}$ and $j\in S$.  
\begin{definition} 
Given the single-agent \ac{hmm} $M_\theta$, for any observation $o$,
the observable operator $\matA^{{\theta}}_{o}$ is a matrix of size $N \times N$ with its $ij$-th entry defined as $
 \matA_{o}^{{\theta}}[i,j] =  \matT_{i, j}^\theta \matO_{o,j} \ ,$
which is the probability of transitioning from state $j$ to state $i$, and at the state $j$, an observation $o$ is emitted.
In matrix form, $
\matA_{o}^\theta = \matT^\theta \text{diag}(\matO_{o, 1}, \dots, \matO_{o, N})$.
\end{definition}

\begin{proposition}
\label{prop:initial_opacity_probability} 

Given the single-agent \ac{hmm} $M_\theta$, the probability of observing $y$ is \begin{equation}
\label{eq:matrix_operation_s0}
\probs_{\theta}(y) = \mathbf{1}_{n}^\top \matA_{o_t}^\theta \dots \matA_{o_0}^\theta \mu_{0}.
\end{equation}
where $\mathbf{1}_{N}$ is a vector of size $N$ containing all $1$. An the derivative of $P_{\theta}(y)$ with respect to $\theta$ is

\begin{equation}
\nabla_{\theta } \probs_{\theta }(y ) = \sum_{i = 0}^t \mathbf{1}_N^\top \matA_{o_t}^\theta \dots \nabla_{\bm{\theta}} \matA_{o_i}^\theta \dots \matA_{o_0}^\theta \mu_{0 }.
\end{equation}
\end{proposition}
\begin{proof}
The proof is based on the property of observable operators \cite{jaeger2000observableoperator} and the derivative of a parameterized tensor product.
\end{proof}

Based on the above analysis for a general single-agent HMM $M(\theta)$ and observation function $E$, we can obtain $\nabla_{\theta_i}  \probs_{\bm{\theta}}(y|\mathbf{T}=i)$ and $\probs_{\bm{\theta}}(y|\mathbf{T}=i)$ by replacing the single-agent HMM with the follower $i$'s HMM $M_i(\theta_i)$ and the observation with $E_i$, and then compute $\nabla_{\theta_i}\probs_{\theta_i}(y)$ and $\probs_{\theta_i}(y)$.
It is observed that   $\nabla_{\theta_j}  \probs_{\bm{\theta}}(y|\mathbf{T}=i)=\bm{0}$ because the observation process of follower $i$ is not influenced by the follower $j$'s policy, when $i\ne j$. 

Using the computation for  $\nabla_{\theta_i} \probs_{\theta_i}(y)$ and $\probs_{\theta_i}(y)$, we can then substitute the corresponding terms in \eqref{eq:HMM_gradient_entropy} to compute $\nabla_{\bm{\theta}} H(\mathbf{T}|Y;\bm{\theta}) $.  It is noted that, though $\mathcal{O}$ is a finite set of observations, it is combinatorial and may be too large to enumerate. To mitigate this issue, we can employ sample approximations to estimate   $\nabla_{\bm{\theta}} H(\mathbf{T}|Y;\bm{\theta})$.

Given $M$ sequences of observations $\{y_1, \dots, y_M\}$, we can approximate $ H(\mathbf{T}|Y;\bm{\theta}) $ by
\begin{equation}
\label{eq:HMM_approx_entropy}
 H(\mathbf{T}|Y;\bm{\theta}) \approx - \frac{1}{M} \sum_{m=1}^M \sum_{k \in \mathcal{T}} \probs_{\bm{\theta}}(k| y_m) \log_2 \probs_{\bm{\theta}}(k | y_m),
\end{equation}
and approximate $\nabla_{\bm{\theta}} H(\mathbf{T}|Y;\bm{\theta})$ by
\begin{equation}
\begin{aligned}
\label{eq:HMM_approx_gradient_entropy}
\nabla_{\bm{\theta}} H(\mathbf{T}|Y;\bm{\theta}) 
&\approx - \frac{1}{M} \sum_{m=1}^M \sum_{k \in \mathcal{T}} \big[ \log_2 \probs_{\bm{\theta}}(k | y_m) \nabla_{\bm{\theta}} \probs_{\bm{\theta}}(k| y_m) \\
& + \probs_{\bm{\theta}}(k| y_m) \log_2 \probs_{\bm{\theta}}(k | y_m) \nabla_{\bm{\theta}} \log_2 \probs_{\bm{\theta}}(y_m) \\
&+ \frac{\nabla_{\bm{\theta}} \probs_{\bm{\theta}}(k | y_m)}{\ln 2} \big].
\end{aligned}
\end{equation}

\paragraph*{Computing $D\bm{\Qstar}(x )$} 
The following proposition~\cite{ma2024adaptive} gives $D Q^\ast(\Ro_i(x_i))$ 
for each follower's MDP $M_i$.

\begin{proposition}[\cite{ma2024adaptive}]
\label{thm:dtheta_dr}
Consider an infinite-horizon \ac{mdp} $M=(S,A,P,s_0,\gamma,R)$ with optimal 
entropy-regularized state-action value function $Q^\ast$. For any 
$(s,a),(\tilde{s},\tilde{a})\in S\times A$:
\begin{equation}
\label{eq:dtheta_dr}
\frac{\partial Q^\ast_{s,a}}{\partial R_{\tilde{s},\tilde{a}}} 
= \mathbf{1}_{(\tilde{s},\tilde{a})}(s,a)
+ \gamma \mathbb{E}_{s'\sim P(\cdot\mid s,a)} 
\sum_{a'\in A} \pi_{Q^\ast}(a'\mid s') 
\frac{\partial Q^\ast_{s',a'}}{\partial R_{\tilde{s},\tilde{a}}},
\end{equation}
where $\mathbf{1}_{(\tilde{s},\tilde{a})}(s,a) = 1$ if $(s,a)=(\tilde{s},\tilde{a})$ 
and $0$ otherwise.
\end{proposition}

For any given $(s',a') \in S \times A$, equation~\eqref{eq:dtheta_dr} is in the form of the standard Bellman equation for state-action value function. This implies that the partial derivative can be computed using any method for solving the Bellman equation~\cite{ma2024adaptive}.

 Lastly, using the aforementioned gradient computation, we can compute the total gradient $DJ(x)$ and employ a gradient-descent algorithm to find a stationary point of  the inference objective function.



        




We now include a complexity analysis for per-step computation in gradient-based method. The computation of $DJ_1(\mathbf{Q}^\ast(x))) $ is polynomial in $|S\times \mathcal{T}|\times T$ where $\mathcal{T}$ is the type set, $T$ is the horizon. The computation of 
$D  \mathbf{Q}^\ast(x) \cdot D\mathbf{R}(x) $ in by $\mathcal{O}(K|\mathcal{T}\times S|^2  |A|^2   + |\mathcal{T}\times S|^3|A|^2)$ where $K$ is the dimension of incentive $x$. Because each step of the computation is polynomial in complexity and the incentive design is performed offline, the proposed method is scalable to large problem instances.

\section{Experiments}
\label{sec:experiments}

We validate the convergence of the algorithm and evaluate its efficiency in distinguishing between different agent types under the special case in section \ref{subsec:full-obs} and general case in section \ref{subsec:partial-obs}.

\subsection{Experiment with a perfectly observable leader and two types of followers with different rewards}
The environment is modeled as an \(8\times 8\) gridworld with two follower types. The initial state is located at \((0,0)\), and the state \((7,7)\) is designated as the target. The followers receive a reward only upon reaching the target state, with \(R_1((7,7)) = 2 R_2((7,7)) = 1\).

Figure~\ref{fig:conv} shows the convergence of the algorithm, and validate the correctness of Lemma~\ref{lem:kl-value-gap}. Figure~\ref{fig:comparison_behavior} demonstrate the divergence of the behavrior of the two types of agents with and without incentive. 

\begin{figure}[h]
    \centering
    \includegraphics[width=0.95\columnwidth]{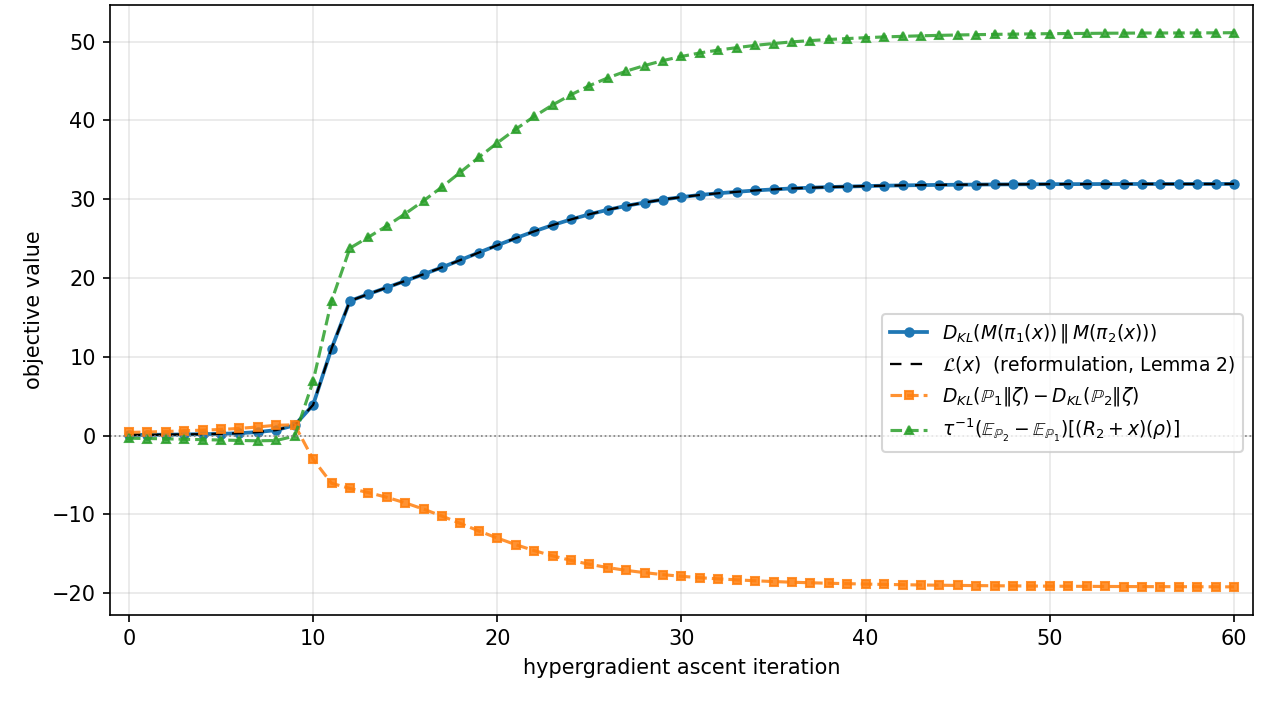}
    \caption{Convergence plot.}
    \label{fig:conv}
\end{figure}

\begin{figure}[h]
    \centering
    \includegraphics[width=0.95\columnwidth]{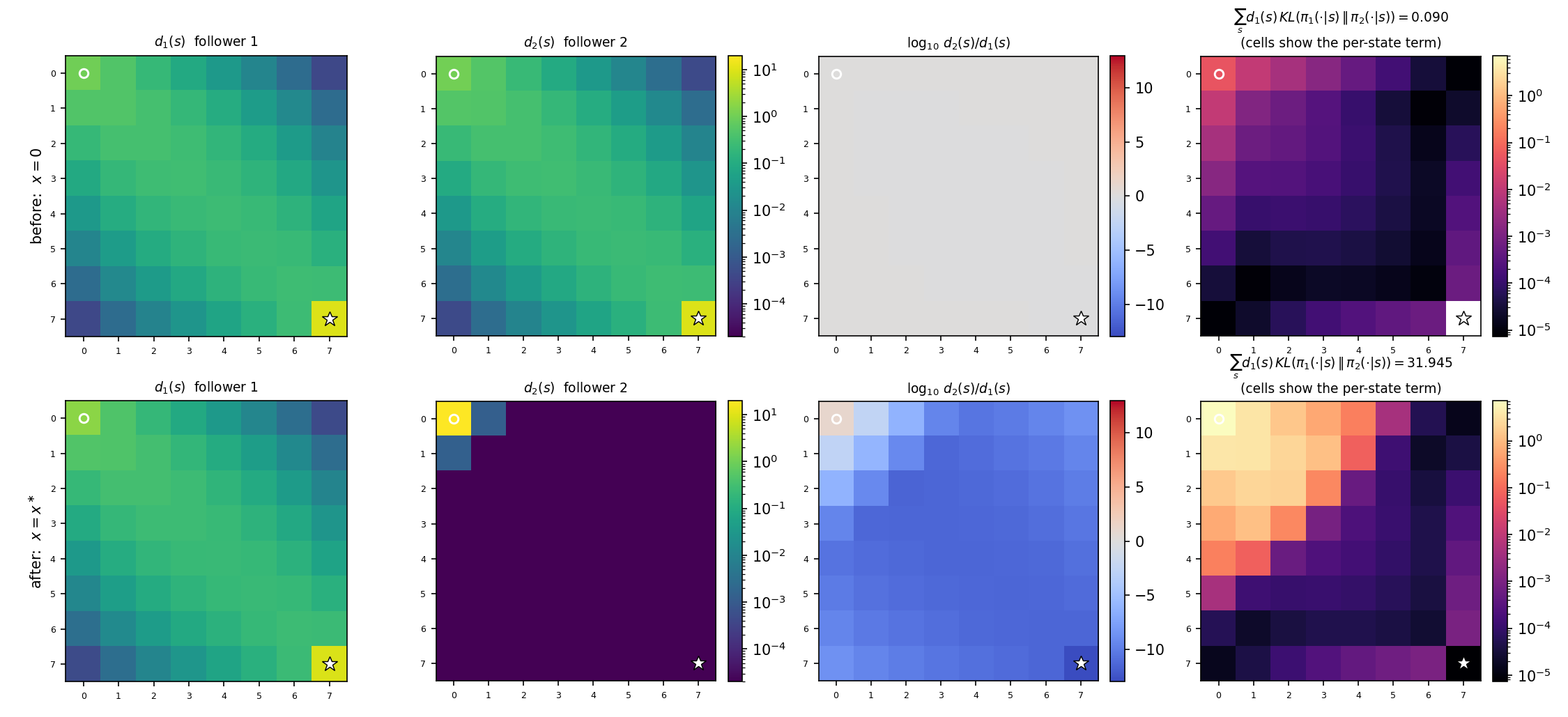}
    \caption{Comparison of Metrics for Behavior Measurement.}
    \label{fig:comparison_behavior}
\end{figure}

Further, we conduct a likelihood ratio test for the leader.

For an observed trajectory prefix $\rho_T$, $\frac{P_1(\rho_T)}{P_2(\rho_T)}
=
\prod_{t=0}^{T-1}
\frac{\pi_1(a_t \mid s_t)}
{\pi_2(a_t \mid s_t)}$. Taking logarithms gives the log-likelihood ratio test statistic $\mathrm{LLR}_T
=
\sum_{t=0}^{T-1}
\log
\frac{\pi_1(a_t \mid s_t)}
{\pi_2(a_t \mid s_t)}$.

The empirical classification accuracy is defined as
\[
\mathrm{Acc}
=
\frac{
\#\{\mathrm{LLR}>0 \mid F_1\}
+
\#\{\mathrm{LLR}<0 \mid F_2\}
+
\frac{1}{2}\#\{\text{ties}\}
}{\#\{\text{sampled trajectories}\}}.
\]


Figure~\ref{fig:inf-acc} shows that \(\mathrm{Acc}\) increases with the iteration number. In addition, longer trajectory lengths \(T\) lead to higher \(\mathrm{Acc}\), indicating that observing more of the trajectory improves the ability to distinguish between the two follower types.

\begin{figure}[h]
    \centering
    \includegraphics[width=0.95\columnwidth]{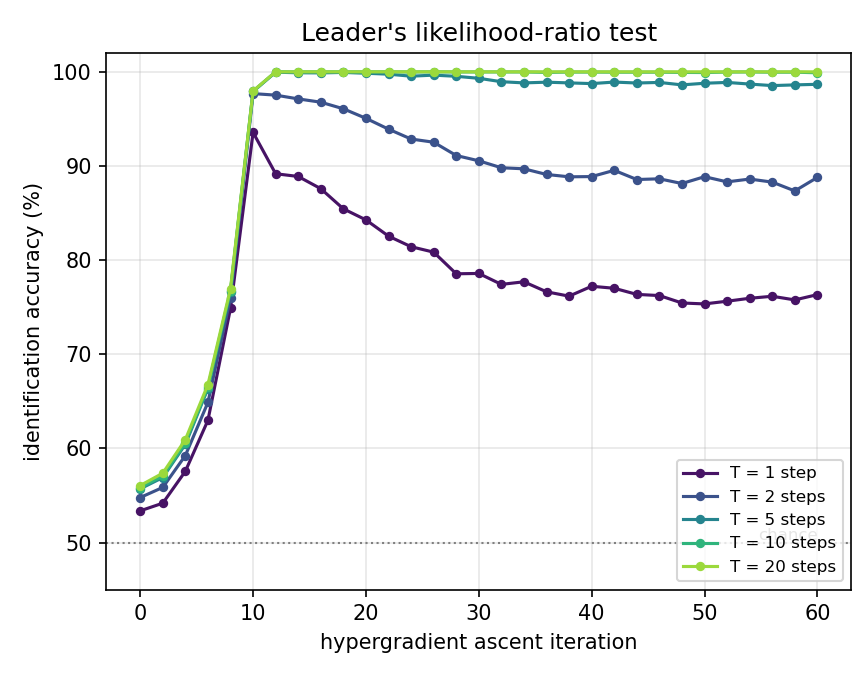}
    \caption{Identification accuracy v.s. iteration numbers.}
    \label{fig:inf-acc}
\end{figure}

In this experiment, we define the feasible set of side payments as $\mathcal{X}=\{x:\|x\|_2\leq B\}$, where \(B\) denotes the incentive budget. We vary \(B\) to investigate how the available incentive budget affects the divergence between the behaviors of the two follower types. The results are shown in Figure~\ref{fig:budget}, where the converged KL divergence increases with \(B\) and gradually saturates as the budget becomes sufficiently large.

\begin{figure}[h]
    \centering
    \includegraphics[width=0.7\columnwidth]{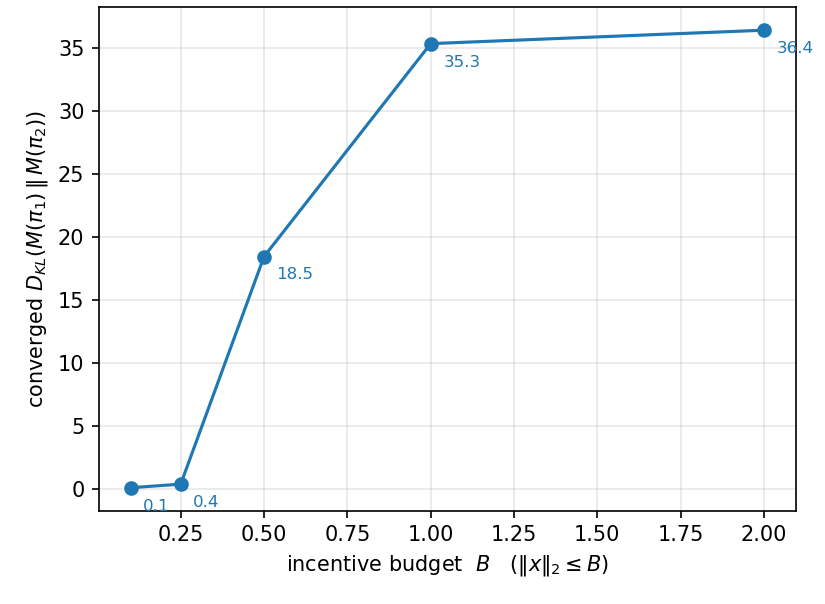}
    \caption{Identification accuracy v.s. iteration numbers.}
    \label{fig:budget}
\end{figure}

\subsection{Experiment with a partially observable leader and multiple types of followers}

We validate the effect of our algorithm within a randomly generated $10 \times 10$ gridworld. The environment configuration includes $12$ targets (marked by red/yellow flags), $20$ obstacles (marked by blocks), and $12$ range sensors with respective field of views (FoVs)  shown in Figure~\ref{fig:gw}. All the targets are sink states. The sensor can detect if an agent is within its FoV, with a  false negative rate $0.05$. The initial state of the agent is randomly sampled from the state space.

\begin{figure}[h]
  \centering
  \begin{minipage}[b]{0.49\columnwidth}
    \centering
    \includegraphics[width=\linewidth]{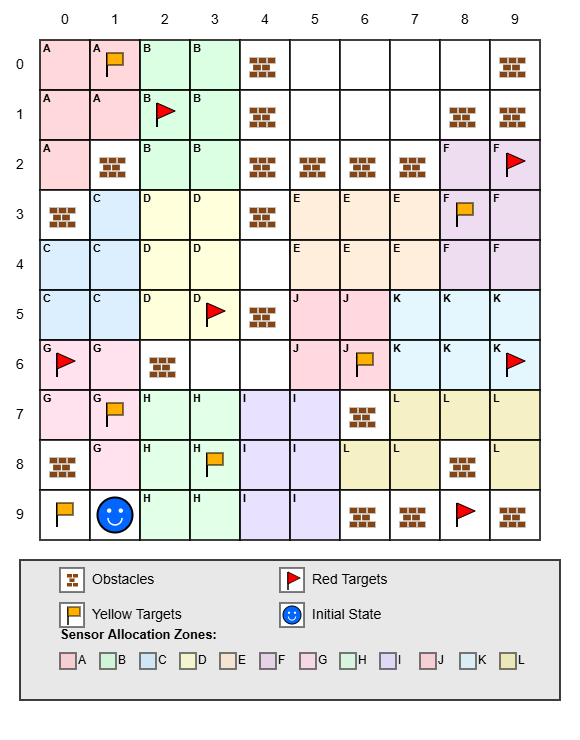}
    \caption{$10 \times 10$ stochastic gridworld with sensor configuration. Sensors' FoVs are shown with different background colors.}
    \label{fig:gw}
  \end{minipage}
  \hfill
  \begin{minipage}[b]{0.49\columnwidth}
    \centering
    \includegraphics[width=\linewidth]{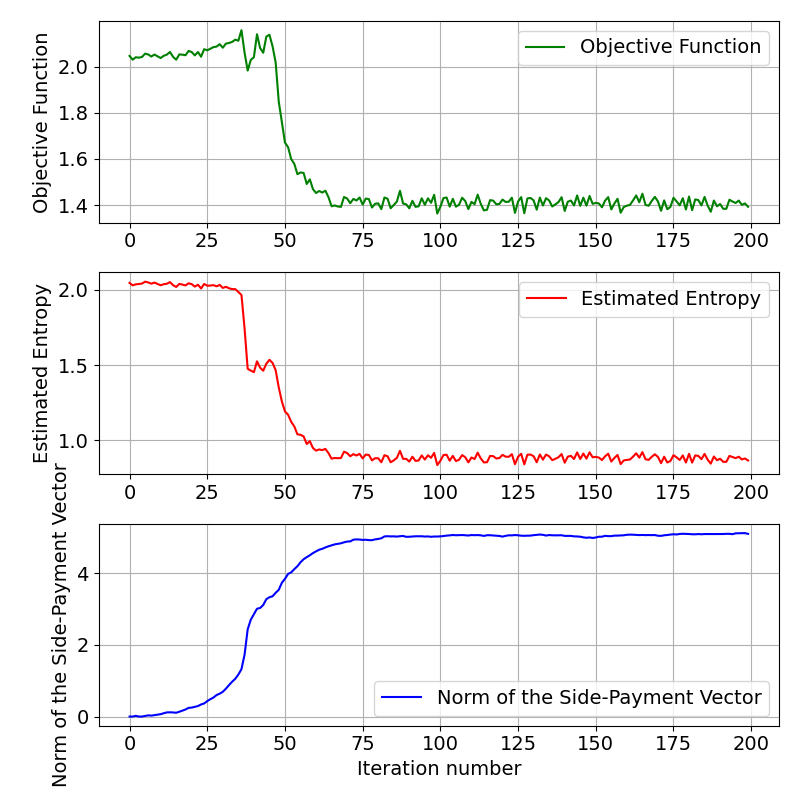}
    \caption{Convergence trend of the value of objective function, estimated entropy and the norm of incentive variable.}
    \label{fig:convergence}
  \end{minipage}
\end{figure}

\paragraph*{Followers' types}  In this gridworld, each agent navigates using four directional actions (north, south, east, west) with stochastic movement dynamics. When attempting to move in a chosen direction, an agent has a  probability $\alpha$ of deviating to one of two adjacent perpendicular directions. We consider five distinct agent types, characterized by different dynamics (varying values for $\alpha$) and reward functions,  as detailed in Table \ref{tab:agent_params}, where   $R_Y$ is the reward for reaching a yellow flag, $R_R$ is the reward for reaching a red flag. A penalty is incurred at each time step before a target is reached. Different agents have different dynamics and different rewards for reaching different targets and per-step penalties. Each type is equally likely to occur.

To balance the inference objective and the cost of incentive design, we set the regularization factor $\lambda=0.1$. We randomly select $20$ states 
\footnote{ The randomly selected states are: $(0,0)$, $(0,5)$, $(1,5)$, $(3,6)$, $(4,1)$, $(4,4)$, $(4,7)$, $(5,9)$, $(6,0)$, $(6,3)$, $(7,4)$, $(7,8)$, $(8,3)$, $(8,6)$, $(8,9)$, $(9,1)$, $(9,2)$, $(9,3)$, $(9,5)$, and $(9,8)$. } where the leader can assign incentive across all state-action pairs for those states. Thus, the size of decision variable $x$ is $80$. The incentive cost function is defined as $h(x)=\sum_{s,a}x_{s,a}$ where the incentive $x_{s,a} \geq 0$ for each sampled state-action pair $(s,a)$.

\begin{table}[h]
  \centering
  \footnotesize
  \caption{Agent type parameters}
  \label{tab:agent_params}
  \begin{tabular}{@{}lrrrr@{}}
    \toprule
    Type & $\alpha$ & $R_\text{Y}$ & $R_\text{R}$ & Per-Step Penalty \\
    \midrule
    0 & 0.05 & 1.0 & 2.0 & -0.5 \\
    1 & 0.20 & 1.0 & 2.0 & -2.0 \\
    2 & 0.05 & 0.1 & 4.0 & 0.0 \\
    3 & 0.15 & 1.0 & 1.0 & -1.0 \\
    4 & 0.00 & 1.0 & 1.0 & -0.1 \\
    \bottomrule
  \end{tabular}
\end{table}

The incentive is initialized to be a zero vector, and the convergence results are shown in Fig.~\ref{fig:convergence}. It is observed that the initial conditional entropy $H(\mathcal{T}| Y_{0:T}; \mathcal{M}(\bm{\pi}^\ast(x_0)))$ with $x_0=\bm{0}$ is  $2.0457$, which is close to the maximal entropy $-\log_2 \frac{1}{5} = 2.32$ derived from a uniform distribution. When the algorithm converges, the total incentive cost $h(x)$ is $5.2590$ and the entropy $H(\mathcal{T}| Y_{0:T}; \mathcal{M}(\bm{\pi}^\ast(x^\ast))) = 0.8652$.

 We compare the inference accuracy to a baseline  which uses maximum likelihood goal recognition  given the observed trajectories in   the original system without incentive. 
To provide a more intuitive understanding about relation between entropy minimization and distinguishability among different types of agents and the benefit of strategic incentive design,  we compare confusion matrices without incentive and with the computed optimal incentive. We sample 100 trajectories for each agent type, compute the posterior probability of type given a trajectory, and sample the predicted agent from that posterior distribution. 
Without incentive, the agents of Types 0, 2, and 4 exhibit similar observation patterns, and thus distinguishing between them is very difficult.  
 Figure \ref{fig:conf-no-x} shows that without incentive, The model tends to classify agents of types 0, 2, and 4 as type 2, while for type 1 agents, the model produces an almost uniform distribution of predictions. In contrast, Figure \ref{fig:conf-x-opt} shows significantly improved true positive rates along the diagonal when using optimal incentive. The average true positive rate over five experiments is 28.73\% (SE 0.0067, 95\% CI 27.00–30.47\%) without incentive and 57.12\% (SE 0.0061, 95\% CI 55.43–58.81\%) with incentive. It is worth noting that these optimal incentive are not solely designed to minimize entropy, as the objective function is also regularized by the cost of the incentive. 
By applying our algorithm to modify the reward functions, we significantly enhance the leader's ability to distinguish between the five types of agents. 

\begin{figure}[H]
    \centering
    \begin{subfigure}[b]{0.48\columnwidth}
        \centering
        \includegraphics[width=\linewidth]{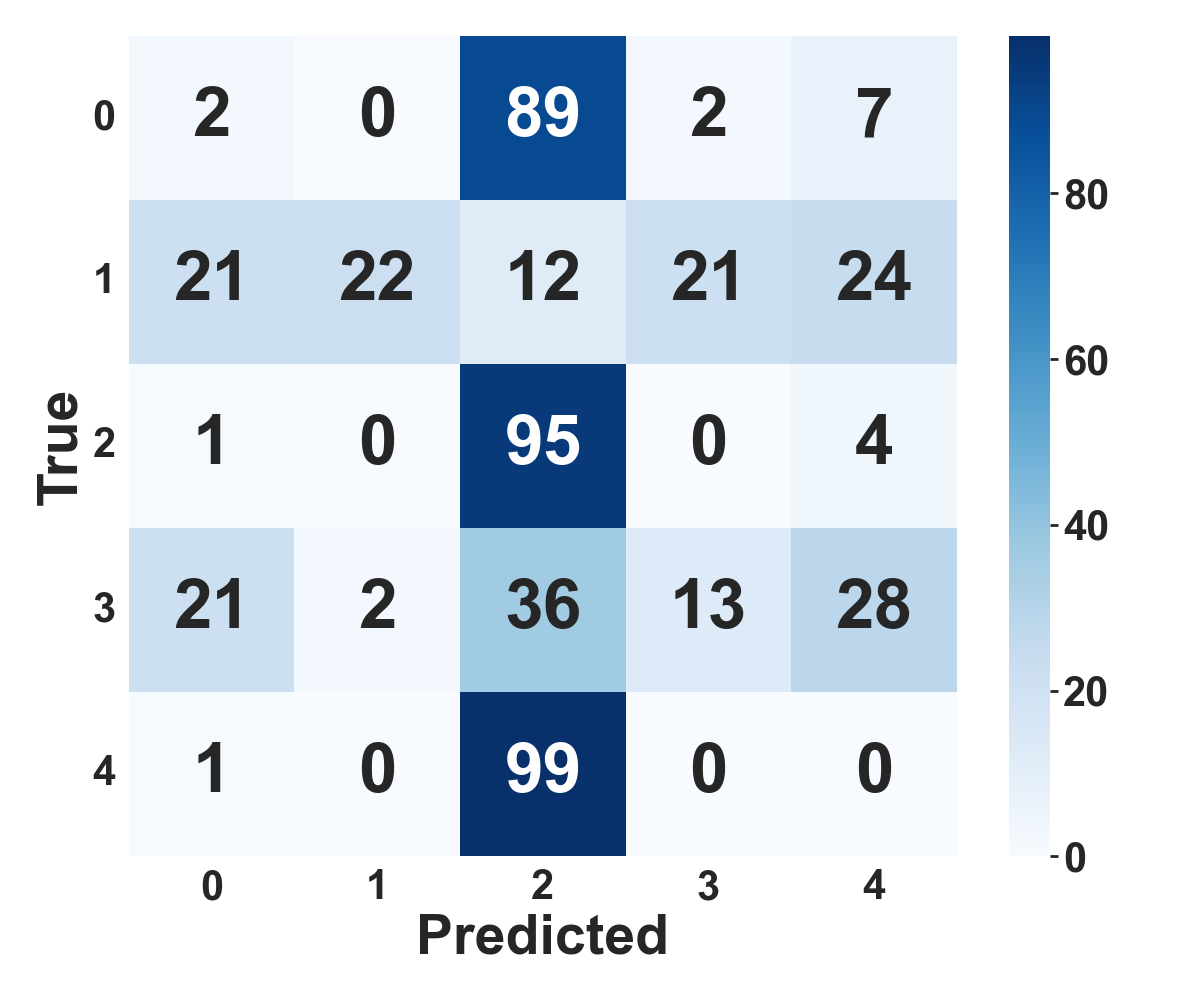}
        \caption{Confusion matrix with no incentive.}
        \label{fig:conf-no-x}
    \end{subfigure}
    \hfill
    \begin{subfigure}[b]{0.48\columnwidth}
        \centering
        \includegraphics[width=\linewidth]{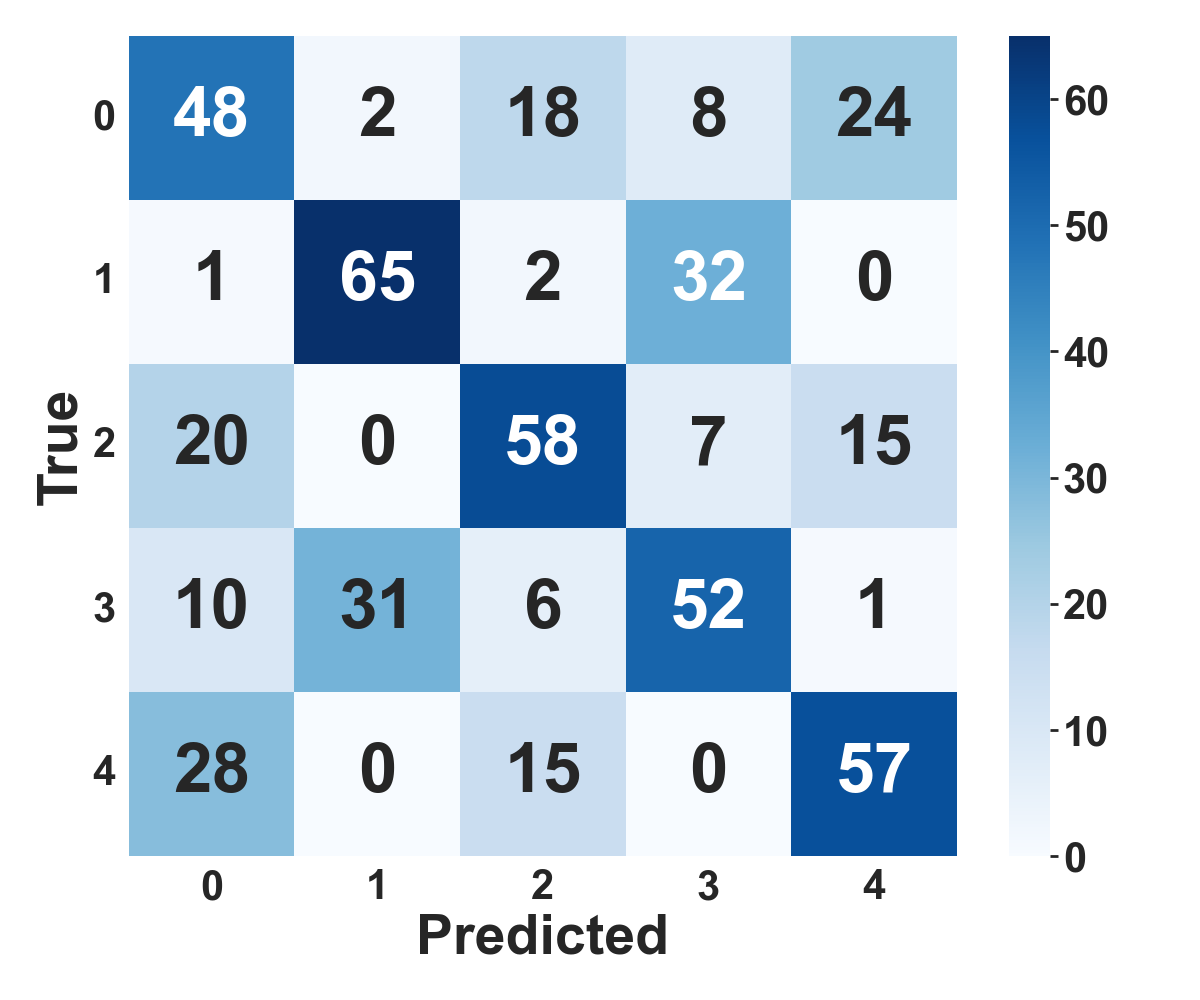}
        \caption{Confusion matrix with the optimal incentive design.}
        \label{fig:conf-x-opt}
    \end{subfigure}    
    \caption{Performance of active inference: comparison of confusion matrices without and with the optimal incentive design.}
    \label{fig:experiment-result}
\end{figure}

\section{Conclusion}
\label{sec:concluison}
 In this paper, we introduce a novel framework for active inference through  incentive design problems. Our main results are presented in two parts. In the first part, we show that, when the leader has  perfect observation and the followers' best responses induce Gibbs trajectory distributions, maximizing the distinguishability between two followers is equivalent to maximizing the cross-reward value gap. In the second part, we extend the framework to settings where the leader has partial observation and there are multiple followers. We develop a gradient-based optimization method to find a  locally optimal solution for the leader's incentive policy. By deploying the incentive policy, the leader can significantly  improve the inference accuracy and efficiency of the unknown follower's type despite partial observations.  We leverage observable operator models to compute the gradient of the conditional entropy with respect to followers' policy parameters.
 
Future work could explore several directions: 1) Integrating active inference with adaptive incentive design for personalized recommendation systems or human-AI alignment problems; 2) Extending the proposed method from model-based to model-free setting, where the leader only has access to collected data from different followers in the past interactions. This extension will be crucial for many practical applications.


\begin{ack}                               
Partially supported by the Roman Senate.  
\end{ack}

\bibliographystyle{plain}        
\bibliography{autosam}           



\appendix
\section{Proof of Lemma \ref{lem:gibbs}}
\label{app:gibbs}
\begin{proof} 
Expanding the KL divergence, the objective \eqref{eq:kl-variational} reads
\begin{align}
  \mathbb{E}_{\mathbb{P}}\bigl[G(\rho)\bigr]
  - \tau\,D_{\mathrm{KL}}(\mathbb{P}\|\zeta)
  &= \mathbb{E}_{\mathbb{P}}\!\left[
       G(\rho) - \tau\log\frac{\mathbb{P}(\rho)}{\zeta(\rho)}
     \right].
  \label{eq:obj-expanded}
\end{align}
Since $\mathbb{P}^*(\rho) = \zeta(\rho)\exp(G(\rho)/\tau)/Z$, taking
logarithms gives
\[
  \log\mathbb{P}^*(\rho)
  = \log\zeta(\rho) + \frac{G(\rho)}{\tau} - \log Z,
\]
which rearranges to the following identity
\begin{equation}
  \label{eq:key-identity}
  G(\rho) - \tau\log\frac{\mathbb{P}(\rho)}{\zeta(\rho)}
  = \tau\log Z - \tau\log\frac{\mathbb{P}(\rho)}{\mathbb{P}^*(\rho)}.
\end{equation}
Substituting \eqref{eq:key-identity} into \eqref{eq:obj-expanded} and
taking the expectation under $\mathbb{P}$, noting that $\log Z$ is
constant, yields
\begin{align*}
  \mathbb{E}_{\mathbb{P}}\bigl[G(\rho)\bigr]
  - \tau\,D_{\mathrm{KL}}(\mathbb{P}\|\zeta)
  &= \tau\log Z
     - \tau\,\mathbb{E}_{\mathbb{P}}\!\left[
         \log\frac{\mathbb{P}(\rho)}{\mathbb{P}^*(\rho)}
       \right] \\
  &= \tau\log Z
     - \tau\,D_{\mathrm{KL}}(\mathbb{P}\|\mathbb{P}^*).
\end{align*}
Since $D_{\mathrm{KL}}(\mathbb{P}\|\mathbb{P}^*) \geq 0$ with equality if and
only if $\mathbb{P} = \mathbb{P}^*$, the objective is
bounded above by $\tau\log Z$, and this bound is attained uniquely at
\[
  \mathbb{P}^*(\rho)
  = \frac{\zeta(\rho)\exp\!\left(G(\rho)/\tau\right)}{Z},
  \quad
  Z \coloneqq \sum_{\rho}\zeta(\rho)\exp\!\left(G(\rho)/\tau\right). 
\]
Setting $\mathbb{P}=\mathbb{P}^*$ (so that $D_{\mathrm{KL}}(\mathbb{P}\|\mathbb{P}^*) = 0$) in the last display gives
$$
\mathbb{E}_{\mathbb{P}^\ast}\bigl[G(\rho)\bigr]
  - \tau\,D_{\mathrm{KL}}(\mathbb{P}^*\|\zeta) = \tau\log Z,
$$
that is,
\begin{equation}
  \label{eq:return-vs-logZ}
  \mathbb{E}_{\mathbb{P}^\ast}\bigl[G(\rho)\bigr]
  = \tau\log Z + \tau\, D_{\mathrm{KL}}(\mathbb{P}^*\|\zeta).
\end{equation}
Note that $\tau \log Z = \mathbb{E}_{\mathbb{P}^\ast}[G(\rho)]$ only when
$D_{\mathrm{KL}}(\mathbb{P}^*\|\zeta) = 0$ (i.e., $\mathbb{P}^*=\zeta$); in
general the log-partition function equals the expected return \emph{net
of} the optimal policy's KL divergence from the reference trajectory
distribution $\zeta$, not the raw expected return itself.
\end{proof}

\section{Proof of Lemma \ref{lem:kl-value-gap}}
\label{app:kl-value-gap}
\begin{proof}
With the Gibbs trajectory distributions obtained from KL-regularized policies, we can derive the KL divergence as follows:
 First,  let $Z(R_i+x)$ be the partition function with respect to reward $R_i+x$, defined by \eqref{eq:gibbs_partition} and $\bbP_i(\cdot) \coloneqq \bbP_i(\cdot, x)$, 
\begin{align*}
    D \coloneqq &D_{\mathrm{KL}}(M(\pi_1(x))\,\|\,M(\pi_2(x)))= \sum_{\rho} \mathbb{P}_1(\rho) \log \frac{\mathbb{P}_1(\rho)}{\mathbb{P}_2(\rho)}\\
    =\;& \sum_{\rho} \mathbb{P}_1(\rho)
        \log \frac{
             \zeta(\rho)/Z(R_1+x)
        }{
             \zeta(\rho)/Z(R_2+x)
        } \\ &+  \sum_{\rho} \mathbb{P}_1(\rho)
        \log \frac{ \exp\!\left(\tfrac{1}{\tau}\sum_{t=0}^{T}\gamma^t(R_1+x)(s_t,a_t)\right)}{ \exp\!\left(\tfrac{1}{\tau}\sum_{t=0}^{T}\gamma^t(R_2+x)(s_t,a_t)\right)},
\end{align*}
where $\tfrac{1}{\tau}\sum_{t=0}^{T}\gamma^t(R_i+x)(s_t,a_t)  = \tfrac{1}{\tau}\sum_{t=0}^{T}\gamma^t(R_i(s_t, a_t)+x(s_t,a_t))$. Since both
followers share the same reference trajectory distribution $\zeta$ (same   MDP dynamics $P$ and reference policy $\pi_0$), the factor $\zeta(\rho)$
 cancels exactly in the log-ratio. Simplifying the remaining exponential terms:
\begin{align*}
   D =\;&\frac{1}{\tau}\,\mathbb{E}_{\mathbb{P}_1}
        \!\left[\sum_{t=0}^{T}\gamma^t(R_1+x-R_2-x)(s_t,a_t)\right]\\
        &
        +\log\frac{Z(R_2+x)}{Z(R_1+x)}.
\end{align*}
By Lemma~\ref{lem:gibbs}, the log-partition function equals the
entropy-regularized optimal value, $\tau \log Z(R_i+x) = V_i^\ast(\mu,
R_i+x)$ --- \emph{not} the expected return under $\mathbb{P}_i$. The two
quantities instead relate as follows (see Appendix):
 \[ \Expect_{\bbP_i}\left[(R_i+x)(\rho)\right] = V_i^{\mathrm{reg},\ast}(\mu, R_i+x) + \tau D_{\mathrm{KL}}(\bbP_i\,\|\,\zeta), \]
 where $(R_i+x)(\rho)   \coloneqq \sum_{t=0}^{T}\gamma^t (R_i(s_t,a_t) +x(s_t,a_t))$.

For  notation simplicity, let $V_i^\ast \coloneqq V_i^{\mathrm{reg},\ast}(\mu, R_i+x)$.
Substituting $\log Z(R_i+x) = V_i^\ast/\tau$ and the identity above for $i=1$
into the expression for $D$, the KL divergence becomes
\begin{align*}
&
D_{\mathrm{KL}}(M(\pi_1(x))\,\|\,M(\pi_2(x )))\\
=& \frac{1}{\tau}\left[ V_1^\ast + \tau D_{\mathrm{KL}}(\bbP_1\,\|\,\zeta) \right] \\
       &-
      \frac{1}{\tau}\, \mathbb{E}_{\mathbb{P}_1 }
       \!\left[\sum_{t=0}^{T}\gamma^t(R_2+x )(s_t,a_t)\right] +\frac1\tau\left(V_2^\ast -V_1^\ast\right) \\
       = &   D_{\mathrm{KL}}(\bbP_1\,\|\,\zeta) + \frac{1}{\tau}\left[ V_2^\ast - \mathbb{E}_{\mathbb{P}_1 }
       \!\left[\sum_{t=0}^{T}\gamma^t(R_2+x )(s_t,a_t)\right]\right].
\end{align*}
The $V_1^\ast$ terms cancel exactly. This intermediate form is not yet the
objective in \eqref{eq:bilevel-gibbs-followers}, since it still references
the log-partition value $V_2^\ast$; we now eliminate it.

To obtain the symmetric form in \eqref{eq:bilevel-gibbs-followers}, apply the identity
above with $i=2$ instead of $i=1$, giving $V_2^\ast =
\mathbb{E}_{\bbP_2}[(R_2+x)(\rho)] - \tau D_{\mathrm{KL}}(\bbP_2\,\|\,\zeta)$.
Substituting this into the last equation,
\begin{align*}
& D_{\mathrm{KL}}(\bbP_1\,\|\,\bbP_2)\\
= &\; D_{\mathrm{KL}}(\bbP_1\,\|\,\zeta) \\
& + \frac{1}{\tau}\Big[\mathbb{E}_{\bbP_2}[(R_2+x)(\rho)] - \tau D_{\mathrm{KL}}(\bbP_2\,\|\,\zeta) \\
&\qquad - \mathbb{E}_{\bbP_1}[(R_2+x)(\rho)]\Big] \\
= &\; D_{\mathrm{KL}}(\bbP_1\,\|\,\zeta) - D_{\mathrm{KL}}(\bbP_2\,\|\,\zeta) \\
& + \frac{1}{\tau}\Big[\mathbb{E}_{\bbP_2}[(R_2+x)(\rho)] - \mathbb{E}_{\bbP_1}[(R_2+x)(\rho)]\Big],
\end{align*}
which is the objective function in \eqref{eq:bilevel-gibbs-followers}.
\end{proof}

\section{Proof of Proposition \ref{prop:pL_px}}
\label{app:pL_px}
The KL divergence is:
\begin{equation}
    D_{\mathrm{KL}}(\mathbb{P}_1 \| \zeta) 
    = \sum_\rho \mathbb{P}_1(\rho;\, x)\log \frac{\mathbb{P}_1(\rho;\, x)}{\zeta(\rho)},
\end{equation}
where only $\mathbb{P}_1(\rho; x)$ depends on $x$; $\zeta$ is fixed.

Differentiating $\sum_\rho \mathbb{P}_1 \cdot \log(\mathbb{P}_1/\zeta)$:
\begin{equation}
    \frac{\partial}{\partial x} D_{\mathrm{KL}} 
    = \underbrace{\sum_\rho \bigl(\frac{\partial}{\partial x} \mathbb{P}_1\bigr)\log\frac{\mathbb{P}_1}{\zeta}}_{\text{Term I}}
    + \underbrace{\sum_\rho \mathbb{P}_1 \cdot \frac{\partial}{\partial x} \log\frac{\mathbb{P}_1}{\zeta}}_{\text{Term II}}.
\end{equation}
Term II vanishes since $\frac{\partial}{\partial x}\log\zeta = 0$ and
$\sum_\rho \mathbb{P}_1 \frac{\partial}{\partial x} \log \mathbb{P}_1 
= \frac{\partial}{\partial x}\sum_\rho \mathbb{P}_1 = 0$.

Term I via the log-derivative trick 
$\frac{\partial}{\partial x} \mathbb{P}_1 = \mathbb{P}_1 \frac{\partial}{\partial x} \log \mathbb{P}_1$:
\begin{equation}
    \frac{\partial}{\partial x} D_{\mathrm{KL}}(\mathbb{P}_1\|\zeta) 
    = \mathbb{E}_{\mathbb{P}_1}\!\left[
        \log\frac{\mathbb{P}_1(\rho)}{\zeta(\rho)} 
        \cdot \frac{\partial}{\partial x} \log \mathbb{P}_1(\rho)
      \right].
\end{equation}

From the Gibbs form 
$\mathbb{P}_1 \propto \zeta\exp\!\left(\tfrac{(R_1+x)(\rho)}{\tau}\right)$:
\begin{align}
    \log\frac{\mathbb{P}_1(\rho)}{\zeta(\rho)} 
    &= \frac{(R_1+x)(\rho)}{\tau} - \log Z(R_1+x), \\[6pt]
    \frac{\partial}{\partial x} \log \mathbb{P}_1(\rho) 
    &= \frac{1}{\tau}\Bigl(
        \frac{\partial}{\partial x} x(\rho) 
        - \mathbb{E}_{\mathbb{P}_1}[\frac{\partial}{\partial x} x(\rho)]
       \Bigr),
\end{align}
where the second line uses 
$\frac{\partial}{\partial x} \log Z = \tfrac{1}{\tau}\mathbb{E}_{\mathbb{P}_1}[\frac{\partial}{\partial x} x(\rho)]$.

The partition function is:
\begin{equation}
    Z(R_1+x) 
    = \sum_\rho \zeta(\rho)\exp\!\left(\frac{(R_1+x)(\rho)}{\tau}\right).
\end{equation}
Differentiating:
\begin{equation}
    \frac{\partial}{\partial x} Z 
    = \sum_\rho \zeta(\rho)
      \exp\!\left(\frac{(R_1+x)(\rho)}{\tau}\right) 
      \cdot \frac{1}{\tau}\frac{\partial}{\partial x} x(\rho).
\end{equation}
Therefore:
\begin{equation}
\begin{aligned}
     \frac{\partial}{\partial x} \log Z 
    &= \frac{\frac{\partial}{\partial x} Z}{Z} 
    = \sum_\rho    \underbrace{\frac{\zeta(\rho)\exp(\cdots)}{Z}}_{=\,\mathbb{P}_1(\rho)} 
      \cdot \frac{1}{\tau}\frac{\partial}{\partial x} x(\rho) 
   \\ & = \frac{1}{\tau}\mathbb{E}_{\mathbb{P}_1}[\frac{\partial}{\partial x} x(\rho)].
\end{aligned}
\end{equation}

\begin{align}
    \frac{\partial}{\partial x} D_{\mathrm{KL}} 
    &= \mathbb{E}_{\mathbb{P}_1}\!\left[
        \left(
          \frac{(R_1+x)(\rho)}{\tau} - \log Z
        \right)
        \cdot
        \frac{1}{\tau}\Bigl(
          \frac{\partial}{\partial x} x(\rho) 
          - \mathbb{E}_{\mathbb{P}_1}[\frac{\partial}{\partial x} x(\rho)]
        \Bigr)
      \right].
\end{align}

The $\log Z$ term vanishes because $\frac{\partial}{\partial x} \log\mathbb{P}_1(\rho)$ is 
zero-mean under $\mathbb{P}_1$. What remains:
\begin{align*}
    &\frac{\partial}{\partial x} D_{\mathrm{KL}} 
    \\= &\frac{1}{\tau^2}\Bigl(
        \mathbb{E}_{\mathbb{P}_1}\!\left[
          (R_1+x)(\rho)\cdot\frac{\partial}{\partial x} x(\rho)
        \right]
        \\&- \mathbb{E}_{\mathbb{P}_1}[(R_1+x)(\rho)]
          \cdot\mathbb{E}_{\mathbb{P}_1}[\frac{\partial}{\partial x} x(\rho)]
       \Bigr).
\end{align*}

\begin{equation}
        \frac{\partial}{\partial x} D_{\mathrm{KL}}(\mathbb{P}_1\|\zeta) 
        = \frac{1}{\tau^2}\,
          \mathrm{Cov}_{\mathbb{P}_1}\!\left[
            (R_1+x)(\rho),\; \frac{\partial}{\partial x} x(\rho)
          \right]
\end{equation}

\begin{equation}
        \frac{\partial}{\partial x} \mathbb{E}_{\mathbb{P}_1}[(R_2+x)(\rho)] 
        = \mathbb{E}_{\mathbb{P}_1}[\frac{\partial}{\partial x} x(\rho)] 
        + \frac{1}{\tau}\mathrm{Cov}_{\mathbb{P}_1}\!\left[(R_2+x)(\rho),\;\frac{\partial}{\partial x} x(\rho)\right]
\end{equation}

\section{Proof of Proposition~\ref{prop:pL_ptheta}} 
\label{app:pL_ptheta}

Let $\mathbb{P}_{\theta_i}$ denote the trajectory distribution induced
by follower $i$'s policy. Throughout this derivation, $x$ is held fixed,
and the reference distribution $\zeta$ is independent of the policy
parameters. We assume differentiability, parameter-independent support
on which $\zeta$ is positive, and sufficient regularity to interchange
differentiation and summation. Define
\begin{equation}
    G_i(\rho;x) \coloneqq (R_i+x)(\rho),
    \qquad
    g_i(\rho) \coloneqq
    \nabla_{\theta_i}\log\mathbb{P}_{\theta_i}(\rho).
\end{equation}
Normalization implies that the score has zero mean:
\begin{equation}
\begin{aligned}
    \mathbb{E}_{\mathbb{P}_{\theta_i}}[g_i(\rho)]
    &=
    \sum_\rho
    \mathbb{P}_{\theta_i}(\rho)
    \nabla_{\theta_i}\log\mathbb{P}_{\theta_i}(\rho)
    \\
    &=
    \nabla_{\theta_i}
    \sum_\rho\mathbb{P}_{\theta_i}(\rho)
    =\mathbf{0}.
\end{aligned}
\end{equation}
This identity holds for arbitrary policy parameters satisfying the
stated assumptions and does not require optimality.

We express the objective as
\begin{equation}
\begin{aligned}
    \mathcal{L}(\theta_1,\theta_2,x)
    &=
    D_{\mathrm{KL}}(\mathbb{P}_{\theta_1}\|\zeta)
    -
    D_{\mathrm{KL}}(\mathbb{P}_{\theta_2}\|\zeta)
    \\
    &\quad+
    \frac{1}{\tau}
    \left(
        \mathbb{E}_{\mathbb{P}_{\theta_2}}[G_2(\rho;x)]
        -
        \mathbb{E}_{\mathbb{P}_{\theta_1}}[G_2(\rho;x)]
    \right)
    -\lambda h(x).
\end{aligned}
\end{equation}

Applying the product rule and the log-derivative identity
$\nabla_{\theta_i}\mathbb{P}_{\theta_i}(\rho)
=\mathbb{P}_{\theta_i}(\rho)g_i(\rho)$ yields
\begin{equation}
\begin{aligned}
    \nabla_{\theta_i}
    D_{\mathrm{KL}}(\mathbb{P}_{\theta_i}\|\zeta)
    &=
    \sum_\rho
    \nabla_{\theta_i}\mathbb{P}_{\theta_i}(\rho)
    \log\frac{\mathbb{P}_{\theta_i}(\rho)}{\zeta(\rho)}
    \\
    &\quad+
    \sum_\rho
    \mathbb{P}_{\theta_i}(\rho)
    \frac{
        \nabla_{\theta_i}\mathbb{P}_{\theta_i}(\rho)
    }{
        \mathbb{P}_{\theta_i}(\rho)
    }
    \\
    &=
    \sum_\rho
    \mathbb{P}_{\theta_i}(\rho)
    \log\frac{\mathbb{P}_{\theta_i}(\rho)}{\zeta(\rho)}
    \,g_i(\rho)
    +
    \underbrace{
        \nabla_{\theta_i}
        \sum_\rho\mathbb{P}_{\theta_i}(\rho)
    }_{=\mathbf{0}}
    \\
    &=
    \mathbb{E}_{\mathbb{P}_{\theta_i}}\!\left[
        \log\frac{\mathbb{P}_{\theta_i}(\rho)}{\zeta(\rho)}
        \,g_i(\rho)
    \right].
\end{aligned}
\end{equation}
Similarly, because $G_2(\rho;x)$ has no explicit dependence on
$\theta_i$,
\begin{equation}
\begin{aligned}
    \nabla_{\theta_i}
    \mathbb{E}_{\mathbb{P}_{\theta_i}}[G_2(\rho;x)]
    &=
    \sum_\rho
    G_2(\rho;x)
    \nabla_{\theta_i}\mathbb{P}_{\theta_i}(\rho)
    \\
    &=
    \mathbb{E}_{\mathbb{P}_{\theta_i}}\!\left[
        G_2(\rho;x)\,g_i(\rho)
    \right].
\end{aligned}
\end{equation}
Moreover, $\nabla_{\theta_i}h(x)=\mathbf{0}$ because $x$ is fixed.

Consequently, the partial gradients at arbitrary policy parameters are
\begin{equation}
    \nabla_{\theta_1}\mathcal{L}
    =
    \mathbb{E}_{\mathbb{P}_{\theta_1}}\!\left[
        \left(
            \log\frac{\mathbb{P}_{\theta_1}(\rho)}{\zeta(\rho)}
            -
            \frac{G_2(\rho;x)}{\tau}
        \right)
        g_1(\rho)
    \right],
\end{equation}
and
\begin{equation}
    \nabla_{\theta_2}\mathcal{L}
    =
    \mathbb{E}_{\mathbb{P}_{\theta_2}}\!\left[
        \left(
            \frac{G_2(\rho;x)}{\tau}
            -
            \log\frac{\mathbb{P}_{\theta_2}(\rho)}{\zeta(\rho)}
        \right)
        g_2(\rho)
    \right].
\end{equation}

We next evaluate these gradients at the optimal policy parameters.
Let $\theta_i^\ast(x)$ denote the parameters of follower $i$'s
optimal policy, and write
\begin{equation}
    \mathbb{P}_i^\ast
    \coloneqq
    \mathbb{P}_{\theta_i^\ast(x)},
    \qquad
    g_i^\ast(\rho)
    \coloneqq
    \left.g_i(\rho)\right|_{\theta_i=\theta_i^\ast(x)}.
\end{equation}
The notation $\left.\cdot\right|_\ast$ denotes evaluation at
$\theta_i=\theta_i^\ast(x)$ for $i=1,2$.

Suppose that the optimal trajectory distributions admit the Gibbs
representation
\begin{equation}
    \mathbb{P}_i^\ast(\rho)
    =
    \frac{
        \zeta(\rho)\exp\!\left(G_i(\rho;x)/\tau\right)
    }{
        Z_i(x)
    },
\end{equation}
where
\begin{equation}
    Z_i(x)
    =
    \sum_\rho
    \zeta(\rho)\exp\!\left(G_i(\rho;x)/\tau\right).
\end{equation}
It follows that
\begin{equation}
    \log\frac{\mathbb{P}_i^\ast(\rho)}{\zeta(\rho)}
    =
    \frac{G_i(\rho;x)}{\tau}
    -
    \log Z_i(x).
\end{equation}

For follower $1$, substituting this identity gives
\begin{equation}
\begin{aligned}
    \left.
    \nabla_{\theta_1}\mathcal{L}
    \right|_\ast
    &=
    \mathbb{E}_{\mathbb{P}_1^\ast}\!\left[
        \left(
            \frac{G_1(\rho;x)-G_2(\rho;x)}{\tau}
            -
            \log Z_1(x)
        \right)
        g_1^\ast(\rho)
    \right]
    \\
    &=
    \frac{1}{\tau}
    \mathbb{E}_{\mathbb{P}_1^\ast}\!\left[
        (R_1-R_2)(\rho)\,g_1^\ast(\rho)
    \right]
    \\
    &=
    \frac{1}{\tau}
    \operatorname{Cov}_{\mathbb{P}_1^\ast}\!\left[
        (R_1-R_2)(\rho),\,g_1^\ast(\rho)
    \right].
\end{aligned}
\end{equation}
The second equality follows because the shared incentive cancels in
$G_1-G_2$, while $\log Z_1(x)$ is independent of $\rho$ and multiplies
a zero-mean score. The covariance representation follows from
$\mathbb{E}_{\mathbb{P}_1^\ast}[g_1^\ast(\rho)]=\mathbf{0}$.

For follower $2$, the same substitution yields
\begin{equation}
\begin{aligned}
    \left.
    \nabla_{\theta_2}\mathcal{L}
    \right|_\ast
    &=
    \mathbb{E}_{\mathbb{P}_2^\ast}\!\left[
        \left(
            \frac{G_2(\rho;x)}{\tau}
            -
            \frac{G_2(\rho;x)}{\tau}
            +
            \log Z_2(x)
        \right)
        g_2^\ast(\rho)
    \right]
    \\
    &=
    \log Z_2(x)\,
    \mathbb{E}_{\mathbb{P}_2^\ast}[g_2^\ast(\rho)]
    \\
    &=\mathbf{0}.
\end{aligned}
\end{equation}

This conclusion also follows directly from follower $2$'s
first-order optimality condition. Define its KL-regularized objective as
\begin{equation}
    J_2(\theta_2,x)
    \coloneqq
    \mathbb{E}_{\mathbb{P}_{\theta_2}}[G_2(\rho;x)]
    -
    \tau D_{\mathrm{KL}}(\mathbb{P}_{\theta_2}\|\zeta).
\end{equation}
The terms in $\mathcal{L}$ that depend on $\theta_2$ equal
$J_2(\theta_2,x)/\tau$. Hence, if $\theta_2^\ast(x)$ is an interior
maximizer of $J_2(\cdot,x)$,
\begin{equation}
    \left.
    \nabla_{\theta_2}\mathcal{L}
    \right|_\ast
    =
    \frac{1}{\tau}
    \left.
    \nabla_{\theta_2}J_2(\theta_2,x)
    \right|_{\theta_2=\theta_2^\ast(x)}
    =
    \mathbf{0}.
\end{equation}
Thus, the vanishing gradient is a stationarity property at the
optimal policy.
\end{document}